\documentclass[11pt]{article}

\usepackage[colorlinks=true, allcolors=blue]{hyperref}
\usepackage[symbol]{footmisc}
\usepackage[utf8]{inputenc}
\usepackage{fullpage}
\usepackage{authblk}
\usepackage{enumitem}
\usepackage{xcolor}

\usepackage{tikz-cd, adjustbox}

\usepackage{amsmath, amsfonts, amsthm}
\usepackage{thmtools, thm-restate}
\usepackage{braket}

\usepackage[colorinlistoftodos]{todonotes}

\declaretheorem{theorem}

\declaretheorem{proposition}
\theoremstyle{definition}
\newtheorem{definition}{Definition}
\newtheorem*{definition*}{Definition}
\newtheorem*{remark}{Remark}

\begin{document}

\title{On establishing learning separations between classical and quantum machine learning with classical data}

\author[1]{Casper Gyurik 
 \thanks{\href{mailto:c.f.s.gyurik@liacs.leidenuniv.nl}{c.f.s.gyurik@liacs.leidenuniv.nl}}}
\author[1]{Vedran Dunjko
\thanks{\href{mailto:v.dunjko@liacs.leidenuniv.nl}{v.dunjko@liacs.leidenuniv.nl}}}
\affil[1]{{\small applied Quantum algorithms (aQa), Leiden University, The Netherlands}} 

\date{\today}

\maketitle

\begin{abstract}
Despite years of effort, the quantum machine learning community has only been able to show quantum learning advantages for certain contrived cryptography-inspired datasets in the case of classical data. 
In this note, we discuss the challenges of finding learning problems that quantum learning algorithms can learn much faster than any classical learning algorithm, and we study how to identify such learning problems.
Specifically, we reflect on the main concepts in computational learning theory pertaining to this question, and we discuss how subtle changes in definitions can mean conceptually significantly different tasks, which can either lead to a separation or no separation at all.
Moreover, we study existing learning problems with a provable quantum speedup to distill sets of more general and sufficient conditions (i.e., ``checklists'') for a learning problem to exhibit a separation between classical and quantum learners.
These checklists are intended to streamline one’s approach to proving quantum speedups for learning problems, or to elucidate bottlenecks.
Finally, to illustrate its application, we analyze examples of potential separations (i.e., when the learning problem is build from computational separations, or when the data comes from a quantum experiment) through the lens of our approach.
\end{abstract}

\section{Introduction}
\label{sec:introduction}
Quantum machine learning (QML) is a bustling field~\cite{arunachalam:qlearning} with the potential to deliver quantum enhancements for practically relevant problems.
The most obvious examples of practical quantum advantage for learning problems occur when the data itself comprises genuine quantum states~\cite{huang:qadvantage}.
In contrast, when the data is classical the community has only been able to establish quantum advantages for cryptography-inspired datasets, which are of limited practical relevance~\cite{liu:dlp, servedio:q_c_learnability}.
In general, an important goal of the community is to find practically relevant learning problems for which one can prove that quantum learners have an advantage over classical learners.
Although the question is of high relevance in a popular area of research, progress here has been slow. 
One of the issues is that proving learning separations is cumbersome, even if one assumes unproven computational separations (e.g., that $\mathsf{BQP} \neq  \mathsf{BPP}$, or that computing expectation values stemming from depth-limited circuits is classically intractable). 
In particular, it is known that access to data can enable a classical computer to compute otherwise intractable functions (see Section \ref{subsubsec:power_data}). 
Another issue is that subtle changes in learning definitions can mean conceptually significantly different tasks, and can either lead to a separation or no separation at all. 
For example, if the objective is to learn a function from a set that contains functions which are known to be hard-to-evaluate for classical computers (even approximately), does this then always constitute a learning separation? 
If we demand that a classical learner has to \emph{evaluate} the learned function, then the answer may be yes, and learning separations may be easier to find. 
However, such an answer conflates computational hardness with hardness of learning, at least in an intuitive sense.
On the other hand, if we only demand the learner to provide a \textit{specification} of the function that underlies the data, then it is not at all clear what the relationship between computational hardness and learning hardness is.
Fortunately, related questions have been studied in classical computational learning theory literature in the 90s~\cite{kearns:clt, kearns:crypto}, which we can build on.

To facilitate progress in the understanding of quantum machine learning separations, in this note we $(i)$ reflect on the main concepts in computational learning theory pertaining to this question by discussing possible definitions of learning separations and their motivations, $(ii)$ discuss the relationship between learning hardness and computational hardness, and $(iii)$ provide two sets of sufficient criteria (i.e., two ``checklists'') that imply a provable
learning separation. 
The first checklist is summarized in Theorem~\ref{thm:recipe_separation}, and it is based on random-verifiable functions to make it applicable when data can be efficiently generated.
The second checklist is summarized in Theorem~\ref{thm:recipe_separation2}, and it is based on additional assumptions involving a number of computational complexity classes with the aim of being applicable to the setting where the data is generated by a quantum experiment.
These checklists can also help identify missing steps in proving learning separations for problems that are potential candidates for exhibiting such separations. 
The results of our work may facilitate the identification of practically relevant tasks where genuine quantum advantages can be formally substantiated.
While much of what we discuss is known to researchers in computational learning theory, since quantum machine learning is interdisciplinary we believe that it is still useful to review the known results, introducing examples from quantum machine learning literature when possible. 

\vspace{-10pt}
\paragraph{Organization of the note}
In Section~\ref{sec:definitions} we recall the definition of efficient learnability and we discuss what exactly constitutes a separation between classical and quantum learners.
Next, in Section~\ref{subsec:dlp} and Section~\ref{subsec:rsa}, we discuss two known examples of separations between classical and quantum learners~\cite{liu:dlp, servedio:q_c_learnability}.
Afterwards, in Section~\ref{sec:recipe}, we highlight the commonalities and differences between what we know quantum computers are good at and what is required for establishing learning separations, and we distill two sets of general and sufficient conditions (i.e., two ``checklists'') to more easily identify whether a problem has sufficient ingredients for a learning separation, or to elucidate the bottlenecks in proving this separation.
These checklists are not general characterizations of all possible learning separation.
Specifically, the first checklist (which is summarized in Theorem~\ref{thm:recipe_separation}) is distilled from the two examples of separations~\cite{liu:dlp, servedio:q_c_learnability}, and the second checklist (which is summarized in Theorem~\ref{thm:recipe_separation2}) is aimed to be applicable to the setting where data comes from a quantum process by involving a number of additional computational complexity classes. 
Finally, in Section~\ref{subsec:recipe_examples}, we provide examples of how our checklists can be used to prove learning separations or elucidate the bottlenecks in proving such a separation.
More precisely, we show that our checklists indeed captures the two separations of~\cite{liu:dlp, servedio:q_c_learnability}, and we elucidate the bottlenecks in proving separations for learning problems which intuitively are promising candidates for potential learning separations (i.e., when the learning problem is build up from a computational complexity separation, or when the data comes from a quantum experiment). 

\section{Definition of a learning separation}
\label{sec:definitions}

In this note we use the standard terminology of the (efficient) \textit{probably approximately correct} (PAC) learning framework, and we focus on the supervised learning setting (for an overview of the generative modelling setting see~\cite{sweke:gen_mod}).
In this framework a learning problem is defined by a family of \textit{concept classes} $\{\mathcal{C}_n\}_{n \in \mathbb{N}}$, where each concept class $\mathcal{C}_n$ consists of \textit{concepts} which are binary-valued functions on an \textit{input space} $\mathcal{X}_n$ (in this note we assume $\mathcal{X}_n$ is either $\{0,1\}^n$ or $\mathbb{R}^n$). 
As input the learning algorithm has access to a procedure $EX(c, \mathcal{D}_n)$ (sometimes called an \textit{oracle}) that runs in unit time, and on each call returns a labeled \textit{example} $(x, c(x))$, where $x \in \mathcal{X}_n$ is drawn according to a \textit{target distribution} $\mathcal{D}_n$.
Finally, the learning algorithm has associated to it a family of hypothesis classes $\{\mathcal{H}_n\}_{n \in \mathbb{N}}$, and its goal is to output a hypothesis $h \in \mathcal{H}_n$ -- which are also binary-valued functions on $\mathcal{X}_n$ --  that is in some sense ``close'' to the concept $c \in \mathcal{C}_n$ generating the examples.
Note that there is also the notion of \emph{proper} PAC learning, where one is only allowed to output a hypothesis from the concept class (i.e., the hypothesis class is the same as the concept class).
In the standard PAC learning framework the learning algorithm has to identify (and/or evaluate) such a good hypothesis using $\mathcal{O}\left(\mathrm{poly}(n)\right)$ many queries to $EX(c, \mathcal{D}_n)$, and the computational complexity (i.e., ``runtime'') of the learning algorithm is not considered. 
In this note however, we focus on the \emph{efficient} PAC learning framework, where the learning algorithm must output such a good hypothesis in \emph{time} $\mathcal{O}\left(\mathrm{poly}(n)\right)$ (note that this also implies that the learning algorithm can only use $\mathcal{O}\left(\mathrm{poly}(n)\right)$ many queries to $EX(c, \mathcal{D}_n)$).
The formal framework of PAC learning deals with the more well-known problem of binary-valued supervised learning. 
Although stated abstractly, the concepts can be thought of as pictures of cats (more precisely, abstract functions which attain value 1 for representations of cats), and the notion of the oracle captures the notion that we are typically given a training set, and no more. 
We formally define efficient PAC learnability as follows.

\begin{definition}[efficient PAC learnability]
\label{def:learnability}
A family of concept classes $\{\mathcal{C}\}_{n \in \mathbb{N}}$ is \textit{efficiently PAC learnable} under target distributions $\{\mathcal{D}_n\}_{n \in \mathbb{N}}$ if there exists a family of hypothesis classes $\{\mathcal{H}_n\}_{n \in \mathbb{N}}$ and learning algorithms $\{\mathcal{A}_n\}_{n \in \mathbb{N}}$ with the following property: for every concept $c \in \mathcal{C}_n$, and for all $0 < \epsilon < 1/2$ and $0 < \delta < 1/2$, if $\mathcal{A}_n$ is given access to $EX(c, \mathcal{D}_n)$ and $\epsilon$ and $\delta$, then $\mathcal{A}_n$ outputs a specification{\color{blue}\footnotemark[1]} of a hypothesis $h \in \mathcal{H}_n$ that with probability at least $1 - \delta$ (over the random examples drawn by calls to $EX(c, \mathcal{D}_n)$ and internal randomization of $\mathcal{A}_n$) satisfies
\[
\mathbb{P}_{x \sim \mathcal{D}_n}\big[h(x) \neq c(x)\big] \leq \epsilon.
\]
Moreover, the learning algorithm $\mathcal{A}_n$ must run in time $\mathcal{O}(n, \mathrm{size}(c){\color{blue}\footnotemark[2]}, 1/\epsilon, 1/\delta)$.
\end{definition}
\footnotetext[1]{The hypotheses (and concepts) are specified according to some enumeration $R: \cup_{n \in \mathbb{N}}\{0,1\}^n \rightarrow \cup_{n} \mathcal{H}_n$ (or, $\cup_{n} \mathcal{C}_n$) and by a ``specification of $h \in \mathcal{H}_n$'' we mean a string $\sigma \in \{0,1\}^*$ such that $R(\sigma) = h$ (see~\cite{kearns:clt} for more details).}
\footnotetext[2]{The runtime of the learning algorithm needs to scale at most polynomially in the size of a concept, which is defined as the size of the \emph{smallest} possible specification of the concept (see~\cite{kearns:clt} for more details).}

If the learning algorithm is a classical algorithm (or, a quantum algorithm), we say that the concept class is \textit{classically efficiently learnable} (or, \textit{quantumly efficiently learnable}, respectively).
It might seem unusual to consider the hypothesis class separately from the learning algorithm.
However, note that in practice one implicitly uses some fixed hypothesis class for the problem at hand.
For instance, in deep learning, the hypothesis class consists of all functions realizable by a deep neural network with some fixed architecture (i.e., enumerated by the weights), and the learning algorithm performs gradient descent to find the optimal weights. 
The hypothesis class used by the learning algorithm can thus be cleverly adapted to the concept class that we want to learn, so we do not want to place unnecessary restrictions on it, though we do not want to leave it completely unconstrained either.
Specifically, it turns out to be pointless if we constrain the learning algorithm to run in polynomial time, but we allow the hypothesis to run for superpolynomial time.
Namely, if we allow hypotheses to run for superpolynomial time, then any concept class that can be learned by a superpolynomial-time learning algorithm, can also be learned by a polynomial-time learning algorithm (i.e., the constraint that the learning algorithm runs in polynomial time is somehow vacuous). 
The reason that this holds comes from the fact that if a concept class $\{\mathcal{C}_n\}_{n \in \mathbb{N}}$ is learnable by a superpolynomial-time learning algorithm using some hypothesis class $\{\mathcal{H}_n\}_{n \in \mathbb{N}}$, then one can construct a new hypothesis class $\{\mathcal{H}'_n\}_{n \in \mathbb{N}}$ (whose concepts are enumerated by all sets of examples) that a polynomial-time learning algorithm can use to learn $\{\mathcal{C}_n\}_{n \in \mathbb{N}}$ {\color{blue}\footnotemark[3]}.
We therefore restrict ourselves to \emph{polynomially evaluatable} hypothesis classes~\cite{kearns:clt}, and we will make the distinction whether they are polynomially evaluatable using a quantum or classical computer.
\footnotetext[3]{See Appendix~\ref{appendix:poly_eval} for more details.}

\begin{definition}[polynomially evaluatable]
    A hypothesis class $\mathcal{H}$ is \emph{polynomially evaluatable} if there exists an evaluation algorithm $\mathcal{A}_{\mathrm{eval}}$ that on input  $x \in \mathcal{X}_n$ together with a specification{\color{blue}\footnotemark[1]} of any hypothesis $h \in \mathcal{H}_n$, outputs $\mathcal{A}_{\mathrm{eval}}(x, h) = h(x)$ in time $\mathcal{O}(n, \mathrm{size}(h))$.
\end{definition}

If the evaluation algorithm is a classical algorithm (or, a quantum algorithm), we say that the hypothesis class is \textit{classically polynomially evaluatable} (or, \textit{quantumly polynomially evaluatable}, respectively).
For example, the hypotheses could be specified by a polynomial-sized Boolean circuit, in which case there is a polynomial-time classical evaluation algorithm (e.g., a classical Turing machine that can simulate Boolean circuits).
On the other hand, the hypotheses could also be specified by a polynomial-depth quantum circuit, in which case there is an efficient quantum evaluation algorithm.
If the family of quantum circuits is universal, then the hypothesis class will be quantumly polynomially evaluatable, but not classically polynomially evaluatable (assuming $\mathsf{BPP}\neq\mathsf{BQP}$).

In some cases, the hypothesis class can be chosen in some particular way for a particular purpose.
Specifically, it can be of practical-relevance to \textit{fix} the hypothesis class used by the learning algorithm.
To give a physics-motivated example, when studying phases of matter one might want to identify what physical properties characterize the phase.
One can formulate this problem as finding a specification of the correct hypothesis selected from a hypothesis class consisting of possible \emph{order parameters}.
More precisely, we fix the hypotheses to be of a very special shape, which compute certain expectation values of ground states given a specification of a Hamiltonian.
In this setting, one might not necessarily want to evaluate the hypotheses, as they require one to prepare the ground state, which is generally intractable (even for a quantum computer).
However, considering hypothesis classes that are not efficiently evaluatable will not trivialize the problem in this setting.
The reason for this is that we cannot apply the construction discussed in Appendix~\ref{appendix:details} (which would render the polynomial-time restriction on the learning algorithm obsolete), since this requires us to be able to change the hypothesis class.

One may assume that there is only one way to define a learning separation.
However, it is in fact more subtle, and there are various definitions that each have operationally different meanings.
In particular, one has to be careful whether we constrain the hypothesis class to be  classically- or a quantumly polynomialy evaluatable.
Specifically, there are four categories of learning problems that we can consider:\ concept classes that are either \textit{classically- or quantumly efficiently learnable}, using an hypothesis class that is either \textit{classically- or quantumly polynomially evaluatable}. 
We denote these categories as $\mathsf{CC}, \mathsf{CQ}, \mathsf{QC}$, and $\mathsf{QQ}$, where the first letter signifies whether the concept class is classically- or quantumly efficiently learnable (i.e., there exists either an efficient classical or quantum learning algorithm), and the second letter signifies whether we use a classically- or quantumly polynomially evaluatable hypothesis class.
These distinctions are \emph{not} about the nature of the data (i.e., we only consider the setting where the examples are classical in nature) as it often occurs in literature, and even on the Wikipedia-page of quantum machine learning.
\begin{definition}[categories of learning problem -- free hypothesis class]\hspace{0pt}
\label{def:cat1}
\begin{itemize}[leftmargin=9pt]
    \item Let $\mathsf{CC}$ denote the set of tuples $\big(\{\mathcal{C}_n\}_{n \in \mathbb{N}},  \{\mathcal{D}_n\}_{n \in \mathbb{N}}\big)$ such that the family of concept classes~$\{\mathcal{C}_n\}_{n \in \mathbb{N}}$ is $\mathsf{\textbf{classically}}$ \textit{efficiently learnable} under target distributions $\{\mathcal{D}_n\}_{n \in \mathbb{N}}$ with a $\mathsf{\textbf{classically}}$ \textit{polynomially evaluatable} hypothesis class.
    
    \item Let $\mathsf{CQ}$ denote the set of tuples $\big(\{\mathcal{C}_n\}_{n \in \mathbb{N}},  \{\mathcal{D}_n\}_{n \in \mathbb{N}}\big)$ such that the family of concept classes~$\{\mathcal{C}_n\}_{n \in \mathbb{N}}$ is $\mathsf{\textbf{classically}}$ \textit{efficiently learnable} under target distributions $\{\mathcal{D}_n\}_{n \in \mathbb{N}}$ with a $\mathsf{\textbf{quantumly}}$ \textit{polynomially evaluatable} hypothesis class.
    
    \item Let $\mathsf{QC}$ denote the set of tuples $\big(\{\mathcal{C}_n\}_{n \in \mathbb{N}},  \{\mathcal{D}_n\}_{n \in \mathbb{N}}\big)$ such that the family of concept classes~$\{\mathcal{C}_n\}_{n \in \mathbb{N}}$ is $\mathsf{\textbf{quantumly}}$ \textit{efficiently learnable} under target distributions $\{\mathcal{D}_n\}_{n \in \mathbb{N}}$ with a $\mathsf{\textbf{classically}}$ \textit{polynomially evaluatable} hypothesis class.
    
    \item Let $\mathsf{QQ}$ denote the set of tuples $\big(\{\mathcal{C}_n\}_{n \in \mathbb{N}},  \{\mathcal{D}_n\}_{n \in \mathbb{N}}\big)$ such that the family of concept classes~$\{\mathcal{C}_n\}_{n \in \mathbb{N}}$ is $\mathsf{\textbf{quantumly}}$ \textit{efficiently learnable} under target distributions $\{\mathcal{D}_n\}_{n \in \mathbb{N}}$ with a $\mathsf{\textbf{quantumly}}$ \textit{polynomially evaluatable} hypothesis class.
    
\end{itemize}
\end{definition}

In Definition~\ref{def:cat1} we constrained the hypothesis class just in terms of the resources required to evaluate them. 
However, as discussed in the example of order parameters, it sometimes makes sense to further constrain and even fix the hypothesis class.
\begin{definition}[categories of learning problem -- fixed hypothesis class]\hspace{0pt}\\
\label{def:cat2}
For a fixed a hypothesis class $H = \{\mathcal{H}_n\}_{n \in \mathbb{N}}$, we define the categories of learning problems:
\begin{itemize}[leftmargin=9pt]
    \item Let $\mathsf{C}_H$ denote the set of tuples $\big(\{\mathcal{C}_n\}_{n \in \mathbb{N}},  \{\mathcal{D}_n\}_{n \in \mathbb{N}}\big)$ such that the family of concept classes~$\{\mathcal{C}_n\}_{n \in \mathbb{N}}$ is $\mathsf{\textbf{classically}}$ \textit{efficiently learnable} under target distributions $\{\mathcal{D}_n\}_{n \in \mathbb{N}}$ with the hypothesis class~$H$.
    
    \item Let $\mathsf{Q}_H$ denote the set of tuples $\big(\{\mathcal{C}_n\}_{n \in \mathbb{N}},  \{\mathcal{D}_n\}_{n \in \mathbb{N}}\big)$ such that the family of concept classes~$\{\mathcal{C}_n\}_{n \in \mathbb{N}}$ is $\mathsf{\textbf{quantumly}}$ \textit{efficiently learnable} under target distributions $\{\mathcal{D}_n\}_{n \in \mathbb{N}}$ with the hypothesis class~$H$.
\end{itemize}
\end{definition}

We remark that our definitions do not talk about the computational tractability of evaluating the concept class, which will be discussed shortly.
We now proceed with a few straightforward observations.
Firstly, since any efficient classical algorithm can be simulated by an efficient quantum algorithm it is clear that $\mathsf{CC} \subseteq \mathsf{CQ}$, $\mathsf{CC} \subseteq \mathsf{QC}$, $\mathsf{CC} \subseteq \mathsf{QQ}$, $\mathsf{CQ} \subseteq \mathsf{QQ}$, $\mathsf{QC} \subseteq \mathsf{QQ}$, and $\mathsf{C}_H \subseteq \mathsf{Q}_H$.
Secondly, if the hypothesis class is quantumly polynomially evaluatable, then it does not matter whether we constrain the learning algorithm to be a classical- or a quantum algorithm.
More precisely, any learning problem that is quantumly efficiently learnable using a quantumly polynomially evaluatable hypothesis class $\{\mathcal{H}_n\}_{n \in \mathbb{N}}$ is also classically efficiently learnable using another quantumly polynomially evaluatable hypothesis class $\{\mathcal{H}'_n\}_{n \in \mathbb{N}}$, and vice versa.
The reason for this is that we can delegate the quantum learning algorithm onto the evaluation of the hypotheses.
To see why this holds, consider a learning problem that is efficiently learnable by a quantum learning algorithm $\mathcal{A}^q$ using a quantumly polynomially evaluatable hypothesis class $\{\mathcal{H}_n\}_{n \in \mathbb{N}}$.
Next, consider the hypothesis class $\{\mathcal{H}'_n\}_{n \in \mathbb{N}}$ whose hypotheses are enumerated by all possible sets of examples, and each hypothesis runs the quantum learning algorithm $\mathcal{A}^q$ on its corresponding set of examples, after which it evaluates the hypothesis that $\mathcal{A}^q$ outputs.
Finally, consider the classical learning algorithm that queries the oracle and outputs a specification of the hypothesis in $\{\mathcal{H}'_n\}_{n \in \mathbb{N}}$ that corresponds to the obtained set of examples.
Note that this \textit{classical} learning algorithm can efficiently learn the original learning problem (which a priori was only quantumly efficiently learnable).
This is summarized in the lemma below, and we provide more details in Appendix~\ref{appendix:details}.

\begin{restatable}{lemma}{CQequalsQQ}
\label{lemma:cq=qq}
$\mathsf{CQ} = \mathsf{QQ}$.
\end{restatable}

We would like to reiterate that it is critical that one is allowed to change the hypothesis class when mapping a problem in $\mathsf{QQ}$ to $\mathsf{QC}$.
Note that if the hypothesis class is fixed (as in the order parameter example and Definition~\ref{def:cat2}), then such a collapse does not happen since we are not allowed to change the hypothesis class.
In particular, in the case of Definition~\ref{def:cat2} it is both allowed and reasonable to fix the hypothesis class to be something that is not (classically or quantumly) polynomially evaluatable.
In this note, we focus on separations in terms of the categories defined in Definition~\ref{def:cat1}, and leave separations in terms of the categories defined in Definition~\ref{def:cat2} for future work. 
Having studied the obvious relationships between the different categories of concept classes, we are now ready to precisely specify what it means for a learning problem to exhibit a separation between classical and quantum learning algorithms.
\begin{definition}[learning separation]
A tuple $L = \big(\{\mathcal{C}_n\}_{n \in \mathbb{N}},  \{\mathcal{D}_n\}_{n \in \mathbb{N}}\big)$ is said to exhibit a 

\begin{itemize}
    \item $\mathsf{CC}/\mathsf{QC}$ separation if $L \in \mathsf{QC}$ and $L \not\in \mathsf{CC}$.
    
    \item $\mathsf{CC}/\mathsf{QQ}$ separation if $L \in \mathsf{QQ}$ and $L \not\in \mathsf{CC}$.
\end{itemize}

\end{definition}

Firstly, note that due the previously listed inclusions of classes any $\mathsf{CC/QC}$ separation directly implies a $\mathsf{CC/QQ}$ separation.
Secondly, one could argue that a slight drawback of $\mathsf{CC}/\mathsf{QQ}$ separations is that they can be less about ``learning proper''.
For instance, consider a concept class made up of a single concept that is classically intractable to evaluate -- even in the presence of data -- yet it can be efficiently evaluated by a quantum algorithm (e.g., based on complexity-theoretic assumptions).
Note that this singleton concept class is quantumly efficiently learnable using a quantumly polynomially evaluatable hypothesis class (i.e., the singleton concept class).
On the other hand, it is not classically efficiently learnable using any classically polynomially evaluatable hypothesis class (since this would violate the classical intractability of evaluating the concept).
However, note that the quantum learner \textit{requires no data} to learn the concept class, so it is hard to argue that this is a genuine learning problem. 
Moreover, such a learning separation is not really introducing new insights under the widely-believed complexity-theoretic assumption that there are classically intractable functions that are efficiently evaluatable by a quantum computer. 
This is different if we only allow the quantum learning algorithm to use hypothesis classes with a polynomial time classical evaluation algorithm, as the above example will not be learnable by such a quantum learning algorithm anymore.
However, note that $\mathsf{CC}/\mathsf{QQ}$ (and $\mathsf{CC}/\mathsf{QC}$) separations could still be possible for concept classes that are efficiently evaluatable.
For instance, there are concept classes that are efficiently evaluatable yet not classically efficiently learnable, such as the class of polynomially-sized logarithmic-depth Boolean circuits~\cite{kearns:clt}.
However, in these well known cases the concept classes are also not quantumly efficiently learnable~\cite{arunachalam:hardness_circuits}, and it remains an open question whether one could restrict these concept classes such that they become quantumly efficiently learnable.

Note that an important difference between $\mathsf{CC}/\mathsf{QQ}$ and $\mathsf{CC}/\mathsf{QC}$ separations is what task you require the quantum algorithm for (i.e., what task needs to be classically intractable yet efficiently doable using a quantum algorithm).
In the case of $\mathsf{CC}/\mathsf{QC}$ separations, one needs to show that one really needs a quantum algorithm to \textit{specify} how one would label unseen examples using a classical algorithm.
On the other hand, in the case of $\mathsf{CC}/\mathsf{QQ}$ separations, one needs to show that you really need a quantum algorithm to correctly \textit{evaluate} (i.e., predict) the labels of unseen examples.

Having discussed the subtleties regarding the choice of definition, we now proceed to study two learning separations between classical and quantum learners.
First, we discuss the discrete logarithm concept class studied in~\cite{liu:dlp}, which to the best of our knowledge exhibits a $\mathsf{CC/QQ}$ separation (i.e., it is unclear whether it also exhibits a $\mathsf{CC/QC}$ separation).
It is important to note that the goal of the authors was different in that they further show that a \emph{general-purpose} quantum learning algorithm (i.e., a quantum kernel method) can learn this concept class.
Afterwards, we discuss a concept class based on the RSA cryptosystem that has been studied in~\cite{servedio:q_c_learnability, kearns:clt}.
This concept class exhibits a $\mathsf{CC/QC}$ separation, as a quantum learning algorithm can efficiently learn it using an hypothesis class that is classically polynomially evaluatable.

\subsection{The discrete logarithm concept class}
\label{subsec:dlp}

In this section we discuss the discrete logarithm concept class studied in~\cite{liu:dlp}.
In this work, the authors prove a $\mathsf{CC/QQ}$ separation for the following \emph{discrete logarithm concept class}.

\begin{definition}[Discrete logarithm concept class~\cite{liu:dlp}]
\label{def:dlp}
Fix an $n$-bit prime number $p$ and a generator $a$ of the (multiplicative) group $\mathbb{Z}_p^*$. We define the \emph{discrete logarithm concept class} as $\mathcal{C}^{\mathrm{DLP}}_n = \{c_{i}\}_{i \in \mathbb{Z}^*_p}$, where we define 
\begin{align}
\label{eq:c_dlp}
     c_i(x) = \begin{cases}+1, & \text{if }\log_a x \in [i, i + \frac{p-3}{2}],\\ -1, & \text{else.} \end{cases}
\end{align}
\end{definition}
\begin{remark}
Here $\log_a x$ denotes the discrete logarithm of $x$ with respect to the generator $a$. That is, $\log_a x$ is the smallest integer $\ell$ such that $a^\ell \equiv x \mod p.$
\end{remark}

The authors of~\cite{liu:dlp} show that a general-purpose quantum learning algorithm (i.e., a quantum kernel method) can efficiently learn $\{\mathcal{C}_n^{\mathrm{DLP}}\}_{n \in \mathbb{N}}$ under the uniform distribution.
The hypothesis class that they use is quantumly polynomially evaluatable, and to the best of our knowledge it is unknown whether this class could be learned efficiently using a classically polynomially evaluatable hypothesis class.
Additionally, they show that under the \emph{Discrete Logarithm Assumption} (a standard assumption in cryptography which states that computing the discrete logarithm even on a $\frac{1}{2} + \frac{1}{\mathrm{poly}(n)}$ fraction of possible inputs is classically intractable) no classical learning algorithm can efficiently learn $\{\mathcal{C}_n^{\mathrm{DLP}}\}_{n \in \mathbb{N}}$ using any classically polynomially evaluatable hypothesis class.
It is also useful to mention that in the case of the discrete logarithm, the weaker assumption that it is hard to compute the discrete logarithm in the worst case already implies that it is hard to evaluate it on this small fraction~\cite{blum:hardness}.
This is due to a so-called \textit{worst-to-average-case reduction}, which we will discuss in more detail in Section~\ref{sec:recipe}.

\begin{theorem}[\cite{liu:dlp}]
$L_{\mathrm{DLP}} = \big( \{\mathcal{C}_n^{\mathrm{DLP}}\}_{n \in \mathbb{N}}, \{\mathcal{D}^U_n\}_{n \in \mathbb{N}})$ exhibits a $\mathsf{CC/QQ}$ separation, where $\mathcal{D}^U_n$ denotes the uniform distribution over $\{0,1\}^n$.
\end{theorem}

To the best of our knowledge, it is unknown whether the discrete logarithm concept class also exhibits a $\mathsf{CC/QC}$ learning separation.
As mentioned before, they way to resolve this is by seeing whether a quantum learning algorithm can still efficiently learn $\{\mathcal{C}_n^{\mathrm{DLP}}\}_{n \in \mathbb{N}}$ using a classically polynomially evaluatable hypothesis class.
In the next section, we will discuss an example of a concept class that does exhibit a $\mathsf{CC/QC}$ separation, since it can efficiently be learned by a quantum learner using a classically polynomially evaluatable hypothesis class.

\subsection{The cube root concept class}
\label{subsec:rsa}


In this section, we discuss a concept class based on the RSA cryptosystem that exhibits a $\mathsf{CC/QC}$ separation.
More specifically, this concept class is quantumly efficiently learnable using a classically polynomially evaluatable hypothesis class, whereas no classical learning algorithm can do so efficiently.
The concept class comes from~\cite{kearns:clt}, but there exist similar concept classes that are also based on the RSA cryptosystem (or more generally, on the hardness of factoring Blum integers{\color{blue}\footnotemark[4]}) which all exhibit a $\mathsf{CC/QC}$ separation (e.g., see~\cite{kearns:crypto}).
The fact that these concept classes exhibit a $\mathsf{CC/QC}$ separation was first observed in~\cite{servedio:q_c_learnability}, though using different terminology.
\footnotetext[4]{$N \in \mathbb{N}$ is a Blum integer if $N = pq$, where $p$ and $q$ are distinct prime numbers congruent to $3 \mod 4$.}

\begin{definition}[Cube root concept class~\cite{kearns:clt}]
\label{def:cube-root}
Fix an $n$-bit integer $N = pq${\color{blue}\footnotemark[5]}, where $p$ and $q$ are two $\lfloor n/2\rfloor$-bit primes such that $\gcd\big(3, (p-1)(q-1)\big) = 1$. We define the \emph{cube root concept class} as $\mathcal{C}_n^{\mathrm{root}} = \{c_{i}\}_{i \in [n]}$, where 
\[
c_i(x) = \text{ the $i$th bit of the binary representation of $f_N^{-1}(x)$},
\]
and the function $f_N:\mathbb{Z}^*_N \rightarrow \mathbb{Z}_N^*$ is given by $f_N(x) = x^3 \mod N$.
\end{definition}
\begin{remark}
By requiring $\gcd\big(3, (p-1)(q-1)\big) = 1$, we ensure that $f_N^{-1}$ exists.
\end{remark}
\footnotetext[5]{In our scenario, the integer $N$ is known to the learner beforehand but $p$ and $q$ are not.}

It is important to note that $f_N^{-1}$ (which is used to construct the concepts) is of the form 
\begin{align}
\label{eq:d}
f_N^{-1}(y) = y^{d^*} \mod N,
\end{align}
for some $d^*$ that only depends on $N${\color{blue}\footnotemark[6]}.
The \emph{Discrete Cube Root Assumption} states that if only given $x$ and $N$, then computing $f_N^{-1}(x)$ is classically intractable (which will prevent a classical learning algorithm from efficiently learning $\{\mathcal{C}_n^{\mathrm{root}}\}_{n \in \mathbb{N}}$).
However, if also given $d^*$, then computing $f_N^{-1}$ suddenly becomes classically tractable.
Using Shor's algorithm a quantum learning algorithm can efficiently compute $d^*$ following the standard attack on the RSA cryptosystem.
Thus, the concept class $\{\mathcal{C}_n^{\mathrm{root}}\}_{n \in \mathbb{N}}$ can be efficiently learned by a quantum learning algorithm using the classically polynomially evaluatable hypothesis class $\{\mathcal{H}_n\}_{n \in \mathbb{N}}$, where
\[
\mathcal{H}_n = \big\{f_{d,i}(x) = \text{$i^\text{th}$ bit of }x^d \mod N\text{ }\big|\text{ }d\in [N],\text{ }i \in [n]\big\},
\]
The evaluation algorithm simply computes the $i$th bit of $x^{d} \mod N$ on input $(d, i)$ and $x$ (which can clearly be done in polynomial time using a classical computer).
It is worthwhile to mention that in this example we see the relevance of how the concepts are specified.
The specifications ``$f^{-1}_{N}$ where $f(x) = x^3$'' and ``$f_N^{-1} = x^{d^*}$'' refer to the same functions, yet computing them is in one case classically tractable, and in the other case classically intractable (under standard complexity-theoretic assumptions).
Note that one still has to learn which bit of $x^{d^*}$ is generating the examples, which implies that the learner thus really requires data.
Finally, under the aforementioned \emph{Discrete Cube Root Assumption}, no classical learner can efficiently learn the concept class $\{\mathcal{C}_n^{\mathrm{DLP}}\}_{n \in \mathbb{N}}$~\cite{kearns:clt}.  
These observations are summarized in the following theorem.
\footnotetext[6]{In cryptographic terms, $d^*$ is the private decryption key corresponding to the public encryption key $e=3$ and public modulus $N$ in the RSA cryptosystem.}

\begin{theorem}[\cite{servedio:q_c_learnability, kearns:clt}]
$L_{\mathrm{root}} = \big( \{\mathcal{C}_n^{\mathrm{root}}\}_{n \in \mathbb{N}}, \{\mathcal{D}^U_n\}_{n \in \mathbb{N}})$ exhibits a $\mathsf{CC/QC}$ separation, where $\mathcal{D}^U_n$ denotes the uniform distribution over $\{0,1\}^n$.
\end{theorem}

\smallskip

In this section, we studied the two known examples of separations between classical and quantum learning algorithms.
First, we studied the discrete logarithm concept class of~\cite{liu:dlp}, which exhibits a $\mathsf{CC/QQ}$ separation (it is unknown whether it also exhibits a $\mathsf{CC/QC}$ separation).
Afterwards, we studied the cube root concept class which exhibits a $\mathsf{CC/QC}$ separation.
This second example shows that it is still possible to have a separation if the quantum learning algorithm is required to use a classically polynomially evaluatable hypothesis class.
In the next section, we dissect these examples to understand what the required ingredients were in proving these learning separations.
In particular, we study where these learning separations actually came from.

\section{Where can learning separations come from?}
\label{sec:recipe}

In this section, we attempt to characterize potential learning separations.
In particular, we study what precisely makes learning separations possible.
We begin by discussing what precisely is needed for a learning separation, and how to make sure the learning problem satisfies these requirements.
To establish a separation we must prove that $(i)$ no classical learning algorithm can learn it efficiently, and $(ii)$ there exists an efficient quantum learning algorithm that can. 
We will focus on $(i)$, since $(ii)$ is a different kind of problem that is about providing an instance of an efficient learner. 

\paragraph{Proving classical intractability of learning}
There are at least three possible ways to prove that a classical learning algorithm cannot efficiently learn a concept class, which we discuss below.

Firstly,  this can be achieved by making sure the concept class is not classically polynomially evaluatable even on a fraction of inputs (this is often called \textit{heuristic hardness}, which we discuss in more details below).
An example of this is the $\mathsf{CC/QQ}$ separation of the discrete logarithm concept class discussed in Section~\ref{subsec:dlp}.
However, a separation where the fact that no classical learner can efficiently learn it comes from just the computational hardness of the concepts may feel unsatisfactory.
In particular, it may feel that such a separation is more about the fact that quantum computers can efficiently evaluate some classically intractable functions (a broadly accepted fact), and not about learning proper (at least in an intuitive sense).

A second way to achieve that no classical learning algorithm can efficiently learn the concept class is to consider concepts that are somehow ``obfuscated''.
More precisely, there can be a specification of the concepts that allows for an efficient classical evaluation algorithm, but in the learning problem the concepts are specified in a different way that does not allow for an efficient classical evaluation algorithm.
An example of this is the $\mathsf{CC/QC}$ separation of the cubic root concept class discussed in Section~\ref{subsec:rsa}, where the concepts have an efficient classical evaluation algorithm if they are specified as ``$f_{d^*, i}(x) =\text{ $i$-th bit of }x^{d^*} \mod N$'', but they cannot be efficiently evaluated on a classical computer if they are specified only as ``the inverse of $f_N(x) = x^3 \mod N$''.

Finally, it could also be possible that there exist learning separations for concepts classes which are specified in a way that allow for an efficient classical evaluation algorithm.
In particular, it could be possible that the hardness really lies in pointing out which of the efficiently evaluatable concepts is generating the examples, which is arguably the most about ``learning proper''.
To the best of our knowledge, no such separation is known, though it is not inconceivable that they exist.
As discussed earlier, the class of polynomially-sized logarithmic-depth Boolean circuits is efficiently evaluatable yet not classically efficiently learnable~\cite{kearns:clt}.
However, these concepts are also not quantumly efficiently learnable~\cite{arunachalam:hardness_circuits}, and it is an open question whether it is possible to restrict this concept class such that it becomes quantumly efficiently learnable.

\paragraph{Challenges in deducing hardness of learning from computational complexity}
In the above paragraphs we highlighted the regimes in which one may expect provable classical intractability of learning, i.e., when concepts are classically intractable or specified in an ``obfuscated'' way. 
In either case, any attempt to actually prove classical intractability of learning may rely on a reduction to some other impossibility, as was the case for the discrete logarithm or the cube root concept class.
Specifically, in these cases efficient learning would imply the capacity to efficiently solve either the discrete logarithm or the cube root problem (which are presumed to be classically intractable).
However, in the attempts to do so it becomes apparent that what is needed for learning separations is often much weaker than the types of no-gos we accept and work on proving in other computational complexity contexts. 
To give a quick intuition, in computational complexity theory, one most often cares about the worst-case complexity, whereas in learning it is all about being correct on only a fraction of inputs, so at least this needs to be relaxed. 

More generally, in attempting to establish reductions of this type, one challenge is that we have to rely on results that state that it is hard to evaluate a certain function (related to the concept class) on just a fraction of inputs.
This is resolved by studying \textit{heuristic complexity}, which we discuss in more details below.
To ensure that the concept class remains quantumly efficiently learnable, we will use concepts that are quantumly polynomially evaluatable, which will result in $\mathsf{CC}/\mathsf{QQ}$ separations.
Additionally, we can also use concepts that are ``obfuscated'' in the sense that when given access to examples only a quantum learning algorithm can deduce a description of the concept that allows for an efficient classical evaluation algorithm, which results in $\mathsf{CC/QC}$ separations.
Secondly, another challenge is the presence of examples, which can radically enhance what can be efficiently evaluated (see Section~\ref{subsubsec:power_data}).
We resolve this by considering concepts that allow for efficient example generation, in which case having access to examples does not enhance what can be efficiently evaluated.

How would one use the observations above to simplify proofs of separations?
The general observations above suggest that it may be possible to further fine-tune the requirements and establish a framework that analyzes a learning problem (i.e., the concept class and the distribution) and tells you something about what kind of separations are possible.
Unfortunately, this work does not succeed in establishing such a general mechanism, but we are able to establish formulations that are more general then all of the two known examples of separations discussed in Section~\ref{subsec:dlp} and Section~\ref{subsec:rsa}.
The proof of the below theorem is deferred to Appendix~\ref{appendix:proof}.

\begin{restatable}[Sufficient conditions for separations based on heuristic hardness]{theorem}{recipe}
\label{thm:recipe_separation}

Consider a family of concept classes $\{\mathcal{C}_n\}_{n \in \mathbb{N}}$ and distributions $\{\mathcal{D}_n\}_{n \in \mathbb{N}}$ (i.e., a learning problem).
Suppose there exists another family of concept classes $\{\mathcal{F}_n\}_{n \in \mathbb{N}}$ and a family of invertible functions $\{g_n: \mathcal{X}_n \rightarrow \mathcal{X}_n\}_{n \in \mathbb{N}}$ such that every $c \in \mathcal{C}_n$ can be decomposed as 
\begin{align}
    c(x) = f(g_n^{-1}(x)), \quad \text{for some $f \in \mathcal{F}_n$}.
\end{align}
Then, if $\{\mathcal{F}_n\}_{n \in \mathbb{N}}$ and $\{g_n\}_{n \in \mathbb{N}}$ satisfy the criteria below, the family of concept classes $\{\mathcal{C}_n\}_{n \in \mathbb{N}}$ exhibits a $\mathsf{CC}/\mathsf{QQ}$ or $\mathsf{CC}/\mathsf{QC}$ separation (depending on which criteria are satisfied) under the family of input distributions $\{\mathcal{D}^g_n\}_{n \in \mathbb{N}}$, where $\mathcal{D}_n^g$ denotes the push-forward distribution $g_n(\mathcal{D}_n)$ (i.e., the distribution induced by first sampling $x \sim \mathcal{D}_n$, and then computing $g_n(x)$).

\bigskip

\noindent \underline{\textbf{Criterion 1} (efficient example generation)}: 
\begin{itemize}
    \item There exists a $\mathcal{O}(\mathrm{poly}(n))$-time randomized classical algorithm that draws a sample $(x, c(x))$, where $x \sim \mathcal{D}^g_n$.
\end{itemize}

\medskip

\noindent \underline{\textbf{Criterion 2} (classical intractability)}:
\begin{itemize}
    \item No classical algorithm can efficiently invert $g_n$ on a $1 - \frac{1}{\mathrm{poly}(n)}$ fraction of inputs. 
    More precisely, there does not exist a classical algorithm $\mathcal{A}$ with runtime $\mathcal{O}\left(\mathrm{poly}(n, 1/\epsilon)\right)$ that satisfies
        \begin{align}
            \mathbb{P}_{x \sim \mathcal{D}_n^g} \big[\mathcal{A}(x) \neq g_n^{-1}(x) \big] \leq \epsilon,
            \label{eq:qalg_recipe}
        \end{align}
    where $\mathcal{D}_n^g$ denotes the push-forward distribution $g_n(\mathcal{D}_n)$  (i.e., the distribution induced by first sampling $x \sim \mathcal{D}_n$, and then computing $g_n(x)$).
    \item For every $x \in \mathcal{X}_n$, there exists subsets $\mathcal{C}'_n(x) \subset \mathcal{C}_n$, $\mathcal{X}'_n(x) \subset \mathcal{X}_n$, and a $\mathcal{O}(\mathrm{poly}(n))$-time classical algorithm $\mathcal{B}$ that maps
    \[
    \big\{c'(y) \mid c' \in \mathcal{C}'_n, \text{ }y\in \mathcal{X}'_n\big\} \mapsto g^{-1}_n(x).
    \]
    \begin{itemize}
        \item In other words, if we can efficiently evaluate a subset of the concepts on a subset of points, then we can use this to efficiently invert $g_n$.
    \end{itemize}
\end{itemize}

\medskip

\noindent \underline{\textbf{Criterion 3} (quantum efficient learnability)}:
\begin{itemize}
    \item For a $\mathsf{CC}/\mathsf{QQ}$ separation we require:
    \begin{itemize}
        \item The problem $\big(\{\mathcal{F}_n\}_{n \in \mathbb{N}}, \{\mathcal{D}_n\}_{n \in \mathbb{N}} \big)$ is in $\mathsf{QQ}$.
        That is, $\{\mathcal{F}_n\}_{n \in \mathbb{N}}$ is quantumly efficiently learnable under $\{\mathcal{D}_n\}_{n \in \mathbb{N}}$ with a classically polynomially evaluatable hypothesis class.
        \item There exists a quantum algorithm $\mathcal{A}$ with runtime $\mathcal{O}(\mathrm{poly}(n, 1/\epsilon))$ that satisfies Eq.~\eqref{eq:qalg_recipe}
    \end{itemize}
    \item For a $\mathsf{CC}/\mathsf{QC}$ separation we require:
    \begin{itemize}
        \item The problem $\big(\{\mathcal{F}_n\}_{n \in \mathbb{N}}, \{\mathcal{D}_n\}_{n \in \mathbb{N}} \big)$ is in $\mathsf{QC}$.
        That is, $\{\mathcal{F}_n\}_{n \in \mathbb{N}}$ is quantumly efficiently learnable under $\{\mathcal{D}_n\}_{n \in \mathbb{N}}$ with a classically polynomially evaluatable hypothesis class.
        \item There exists a polynomial time classical evaluation algorithm $\mathcal{A}_{\mathrm{eval}}$, and a quantum algorithm $\mathcal{A}^q$ with runtime $\mathcal{O}(\mathrm{poly}(n, 1/\epsilon))$ that computes a specification $\mathcal{A}^q(g^{-1}_n) = [g^{-1}_n]$ such that the algorithm $\mathcal{A}(x) = \mathcal{A}_{\mathrm{eval}}([g^{-1}_n], x)$ satisfies Eq.~\eqref{eq:qalg_recipe}.
    \end{itemize}
\end{itemize}

\smallskip 

\end{restatable}

The main idea behind the above theorem is that we split our concepts in two parts, i.e., a learning part $\{\mathcal{F}_n\}_{n \in \mathbb{N}}$ and  a classically intractable part $\{g^{-1}_n\}_{n \in \mathbb{N}}$. 
This servers two purposes: on the one hand we ensure quantum learnability by using $\{\mathcal{F}_n\}_{n \in \mathbb{N}}$, and on the other hand we ensure classical intractability of learning by using $\{g^{-1}_n\}_{n \in \mathbb{N}}$.
Specifically, for the classical intractability of learning, we decompose our concepts in such a way that by Criterion 2 we can conclude that the capacity to learn the family of concept classes enables one to efficiently evaluate $\{g^{-1}_n\}_{n \in \mathbb{N}}$.

Let us now describe how parts of the above theorem's conditions correspond to the aforementioned problems in establishing separations.
Firstly, to overcome the challenge that training examples can radically enhance what a classical learner can evaluate, we require that examples are efficiently generatable (Criterion~1).
In particular, if the examples are efficiently generatable, then having access to examples does not enhance what a learner can evaluate relative to a non-learning algorithm, since any non-learning algorithm can first generate examples to put itself on equal footing with a learner.
In other words, if the examples are efficiently generatable, then the existence of an efficient classical learning algorithm implies the existence of an efficient classical algorithm for evaluating the concepts at hand.
In our theorem, we describe what it means for examples to be efficiently generatable in the most general way.
Specifically, we require that the concepts must be \textit{random verifiable} functions~\cite{arrighi:blind} (the definition of which is equivalent to Criterion 1).
However, there are properties that imply random verifiability that are perhaps more intuitive (although more general situations are possible).
For instance, in the existing separations in literature, random verifiability is ensured by letting $g_n$ be a bijection, and by having both $\{\mathcal{F}_n\}_{n \in \mathbb{N}}$ and $\{g_n\}_{n \in \mathbb{N}}$ be efficiently evaluatable. 
Namely, in this case the pair $(g_n(x), f(x))$ for a uniformly random $y$ can be understood as a pair $(x, f(g_n^{-1}(x)))$, where $x$ is chosen uniformly at random (since $g_n$ is a bijection). 
Note that this is a stronger requirement than being random verifiable, since there exist random verifiable functions that may not be efficiently computable classically (e.g., if the graph isomorphism problem is not solvable in polynomial time on a classical computer~\cite{arrighi:blind}).

Having dealt with the challenge that training examples can radically change what a learner can evaluate, we still have to actually separate a classical learner from a quantum learner.
We achieve this by considering concepts that are build up from functions $g_n$ that can be inverted by a quantum algorithm, yet no classical algorithm can do so even on a fraction of inputs. 
As mentioned before, in complexity theory one typically studies the worst-case hardness of a problem, whereas in the context of learning one is only concerned with being correct on a fraction of inputs.
We therefore consider the less widely-known areas of complexity theory that study \emph{average-case} or \emph{heuristic} complexity~\cite{bogdanov:average} (these complexity classes have however been important in connection to cryptography, and, in more recent times in studies involving quantum computational supremacy~\cite{harrow:supremacy}).
In heuristic complexity one considers algorithms that err on a inverse-polynomial fraction of inputs, which matches the PAC-learning requirements on the hypothesis that a learner has to output.
Throughout this note we talk about functions of which some are not binary-valued, in which case they are technically not decision problems.
We admit a slight abuse of notation and actually talk about a more general class of functions with the property that if those are evaluatable, then they allow us to solve a corresponding decision problem (e.g., deciding if the first bit of the function is zero).

\begin{definition}[Heuristic complexity~\cite{bogdanov:average}]
\label{def:heurbpp}
A distributional problem $(L, \{\mathcal{D}_n\})${\color{blue}\footnotemark[7]} is in $\mathsf{HeurBPP}$ if there exists a $\mathrm{poly}(n, 1/\epsilon)$-time randomized classical algorithm $\mathcal{A}$ such that for all $n$ and $\epsilon > 0$:
\begin{align}
\label{eq:heur}
     \mathsf{Pr}_{x \sim \mathcal{D}_n}\Big[\mathsf{Pr}\big(\mathcal{A}(x) = L(x)\big) \geq \frac{2}{3} \Big] \geq 1 - \epsilon,
\end{align}
where the inner probability is taken over the internal randomization of $\mathcal{A}$.\\
\indent Also, a distributional problem $(L, \{\mathcal{D}_n\})$ is in $\mathsf{HeurBQP}$ if there exists a $\mathrm{poly}(n, 1/\epsilon)$-time quantum algorithm $\mathcal{A}$ that satisfies Eq.~\eqref{eq:heur}.
\end{definition}
There are multiple ways to define the notion of average-case complexity. 
In basic textbooks it is often defined in terms of the average-case runtimes of the algorithms (i.e., when averaging according to some distribution over the inputs).
However, this definition can be surprisingly problematic. 
A more straightforward way to define average-case complexity, and the one we utilize, considers algorithm that can never err, but that are allowed to output ``don't know'' on a fraction of inputs. 
We refer the reader to Section 2.2 of~\cite{bogdanov:average} for other various possible definitions.
\begin{definition}[Average-case complexity~\cite{bogdanov:average}]
\label{def:avgbpp}
A distributional problem $(L, \{\mathcal{D}_n\})${\color{blue}\footnotemark[7]} is in $\mathsf{AvgBPP}$ if there exists a $\mathrm{poly}(n, 1/\epsilon)$-time randomized classical algorithm $\mathcal{A}$ such that for all $n$ and $\epsilon > 0$:
\begin{align}
\label{eq:avg1}
    \mathsf{Pr}\big(\mathcal{A}(x) \in \{L(x), \bot{\color{blue}\footnotemark[8]}\}\big) \geq \frac{2}{3}.
\end{align}
where the probability is taken over the internal randomization of $\mathcal{A}$, and 
\begin{align}
\label{eq:avg2}
    \mathsf{Pr}_{x \sim \mathcal{D}_n}\Big[\mathsf{Pr}\big(\mathcal{A}(x) = L(x)\big) \geq \frac{2}{3} \Big] \geq 1 - \epsilon,
\end{align}
where the inner probability is taken over the internal randomization of $\mathcal{A}$.\\
\indent Also, a distributional problem $(L, \{\mathcal{D}_n\})$ is in $\mathsf{AvgBQP}$ if there exists a $\mathrm{poly}(n, 1/\epsilon)$-time quantum algorithm $\mathcal{A}$ that satisfies Eq.~\eqref{eq:avg1} and Eq.~\eqref{eq:avg2}.
\end{definition}
\footnotetext[7]{Here $L:\{0, 1\}^* \rightarrow \{0,1\}$ is a language and $\mathcal{D}_n$ is a distribution over inputs $\{0,1\}^n$.}
\footnotetext[8]{The symbol $\bot$ corresponds to the algorithm outputting ``don't know''.}

For our purpose, we require that there does not exist an efficient randomized classical algorithm that can invert $g_n$ on an $\epsilon$-fraction of the inputs (i.e., it can err).
At first glance, this corresponds to heuristic complexity.
However, in certain settings this also corresponds to average-case complexity.
Specifically, recall that a setting that allowed for efficient example generation is one where $g_n$ is efficiently evaluatable classically. 
In this setting, the hardness of inverting $g_n$ on an $\epsilon$-fraction of the inputs also falls into average-case complexity, since we can turn any algorithm that errs into an algorithm that outputs ``don't know'' by computing $g_n$ and checking if the output is correct.

While heuristic-hardness statements are not as common in quantum computing literature, the cryptographic security assumptions of RSA and Diffie-Hellman are in fact examples of such heuristic-hardness statements{\color{blue}\footnotemark[9]}.
Such heuristic-hardness statements are typically obtained following \emph{worst-case to average-case reductions}, from which it follows that being correct on a fraction of inputs is at least as hard as being correct on all inputs.
For example, we know that computing the discrete logarithm on a $\frac{1}{2} + \frac{1}{\mathrm{poly}(n)}$ fraction of inputs is as hard as computing it for all inputs due to a worst-case to average-case reduction by Blum and Micali~\cite{blum:hardness} (i.e., if there exists an efficient algorithm for a $\frac{1}{2} + \frac{1}{\mathrm{poly}(n)}$ fraction of inputs, then there exists an efficient algorithm for all inputs).

\footnotetext[9]{In fact, most often these assumptions also state that there are also no efficient heuristic algorithms even when given training examples.}

We remark that it is possible to restate learning separations using just terminology of complexity classes.
Specifically, a learning separation combines two types of complexity classes.
As mentioned earlier, the fact that a learner can err on a fraction of inputs is captured by heuristic complexity classes.
Additionally, the fact that data can radically enhance what a leaner can evaluate is captured by the ``sampling advice'' complexity classes defined in~\cite{huang:power}.
Thus, one can reformulate learning separations in complexity-theoretic terms as follows:\ learning separations are exhibited by problems inside $\mathsf{HeurBQP}$ with sampling advice, that do not lie in $\mathsf{HeurBPP}$ with sampling advice.
Sampling advice is closely-related to the notion of ``advice strings'' studied in complexity classes such as $\mathsf{P/poly}$~\cite{huang:power} (i.e., the class of problems that are solvable by a polynomial time classical algorithm using a polynomial sized ``advice string'' which only depends on the size of the input).

Recall that in Theorem~\ref{thm:recipe_separation} we decompose the concepts into two parts:\ a learning part $\{\mathcal{F}_n\}_{n \in \mathbb{N}}$ (ensuring efficient quantum learnability) and a classically intractable part $\{g_n^{-1}\}_{n \in \mathbb{N}}$ (ensuring classical intractability of learning).
Note that separations could in principle be possible without this decomposition.
For instance, there might exist a classically polynomially evaluatable concept classes that exhibits a learning separation (in which case there will not be a classically intractable part).
Moreover, it is also possible to generalize Theorem~\ref{thm:recipe_separation}.
In particular, it would already suffice to require that $(i)$ examples are efficiently generatable, $(ii)$ $\{g_n\}_{n \in \mathbb{N}}$ is heuristically hard to invert, and $(iii)$ $\{\mathcal{C}_n\}_{n \in \mathbb{N}}$ is $\mathsf{QQ}$ or $\mathsf{QC}$ learnable.
More precisely, in contrast to Theorem~\ref{thm:recipe_separation}, the efficient quantum learnability does not necessarily have to come from the efficient learnability of $\{\mathcal{F}_n\}_{n \in \mathbb{N}}$, or the efficient quantum algorithm for $\{g_n^{-1}\}_{n \in \mathbb{N}}$.
However, in Theorem~\ref{thm:recipe_separation} we choose to list in more detail specific criteria that guarantee efficient quantum learnability (i.e., we show what efficient quantum learnability can look like).
Another possible generalization would be to consider concept classes that can be decomposed into smaller concept classes that each have their own decomposition as in Theorem~\ref{thm:recipe_separation}.
These concept classes will still be classically hard to learn, though the efficient quantum learnability has to be guaranteed by something else.
We also remark that Theorem~\ref{thm:recipe_separation} does not capture the setting where the hypothesis class used by the learning algorithm is fixed (e.g., the order-parameter example discussed in Section~\ref{sec:definitions}).
Finally, it is also important to consider what happens when we drop the criteria that examples are efficiently generatable, since this allows us to also be applicable to the important setting where the data is generated by a quantum process.

\paragraph{Learning separations without efficient data generation}
Throughout the above, we focused on the setting where example generation is efficient, as this makes the proof of separation more straightforward by ensuring that access to examples does not help the classical learner.
With regard to the hardness of learning, the hardness of generating examples is a double-edged sword.
On one hand, if generating examples is hard, then this intuitively makes the function even harder for a classical learner.
On the other hand, for the same reason having such hard to generate examples might give more leverage to a learner the harder the evaluation is (see Section~\ref{subsubsec:power_data}). 

From the perspective of quantum machine learning requiring that examples are efficiently generatable classically is a serious problem.
In particular, it is often argued that quantum learning advantages should arise in cases where the examples are generated by a quantum process (i.e., when the true classifier is a classically intractable quantum function).
However, in these cases efficient example generation is generally not possible using a classical computer.
Our results therefore arguably say almost nothing about many relevant scenarios in quantum machine learning, and furthermore leave open the more fundamental question whether efficient example generation is in fact necessary for separations.
We will now address this question in the setting where the data is generated by a quantum process, and analyze the possible scenarios.

Recall that to achieve a learning separation we need both $(i)$~intractability of classical learning, and $(ii)$~efficient quantum learnability.
In Theorem~\ref{thm:recipe_separation}, we considered sufficiently hard functions in order to satisfy~$(i)$, but not so hard as to violate~$(ii)$, while at the same time dealing with the problem of examples offering leverage to classical learners.
It is however possible to satisfy both~$(i)$ and deal with the problem of examples giving leverage to classical learners, by considering even harder functions, and then deal with the issue of efficient quantum learnability separately.
Specifically, we can consider functions that are outside $\mathsf{HeurP}/\mathsf{poly}$ (note that $\mathsf{BPP} \subseteq \mathsf{P/poly}$, so we do not need to worry about randomness as a resource).
Recall that $\mathsf{HeurP}/\mathsf{poly}$ is the class of functions that are heuristically evaluatable (i.e., on an inverse-polynomial fraction of inputs) in polynomial time by a classical algorithm using a polynomial sized ``advice string'' which is input independent, but size dependent.
Note that functions which are outside $\mathsf{HeurP}/\mathsf{poly}$ are not efficiently classically learnable, since if they were then the training examples can be turned into an ``advice string'' that enables one to heuristically evaluate the function (see also Chapter 2 of the Supplementary Information of~\cite{huang:power} for a more formal discussion).

Functions outside $\mathsf{HeurP}/\mathsf{poly}$ (ensuring classical intractability of learning) but inside $\mathsf{BQP}${\color{blue}\footnotemark[10]} (ensuring efficient quantum learnability) are thus prime candidates for learning separations.
Moreover, such functions do not have to be random verifiable, in which case they apply to the setting where the data comes from a quantum experiment.
With this point of view, we now analyze all the possible scenarios for which our previous theorem or other observations elucidate whether learning separations can exist.
Specifically, we need to consider the possible relationships between being in the class $\mathsf{BQP}\backslash \left(\mathsf{HeurP/poly}\right)${\color{blue}\footnotemark[11]}, or the class $\mathsf{HeurP/poly} \cap \mathsf{BQP}$, and being random verifiable.
First, suppose all functions in $\mathsf{BQP}\backslash \left(\mathsf{HeurP/poly}\right)$ are in fact random verifiable.
In this case, Theorem~\ref{thm:recipe_separation} is applicable to these scenarios as well.
However, it is also possible (and arguably, more likely{\color{blue}\footnotemark[12]}) 
that there exist functions $f$ in $\mathsf{BQP}$ that are neither in $\mathsf{HeurP/poly}$ nor random verifiable.
In this case, we can prove the existence of learning separations.
More precisely, we can consider the concept class $\{f\}$, which is trivially learnable by a quantum learner (since $f \in \mathsf{BQP}$), and it is not classically learnable (since $f \not \in \mathsf{HeurP}/\mathsf{poly}$).
We state this observation in the proposition below.
\footnotetext[10]{We choose $\mathsf{BQP}$ here instead of $\mathsf{HeurBQP}$ since it feels more natural when the data comes from a quantum process (i.e., a quantum process is always correct on the function that it implements) and it allows us later on to more easily establish quantum learnability in Theorem~\ref{thm:recipe_separation2}.}
\footnotetext[11]{Here $A\backslash B$ denotes the set $\{c \mid c \in A,\text{ }c \not\in B\}$ (i.e., the set obtained by removing the elements of $B$ from $A$).}
\footnotetext[12]{The converse of this would imply that all $\mathsf{BQP}$-complete functions -- which are also unlikely to be in $\mathsf{HeurP/poly}$, as we believe that factoring or the discrete logarithm are not in $\mathsf{HeurP/poly}$ (since the converse would imply certain breaks in Diffie-Helman or RSA cryptosystems) -- are random verifiable, which would be quite revolutionary in for instance Monte Carlo simulations of quantum systems.} 

\begin{proposition}
\label{prop:norv}
If there exists an $f \in \mathsf{BQP}$  with $f \not \in \mathsf{HeurP}/\mathsf{poly}$ that is \textit{not} random verifiable, then there exists a learning problem which exhibits a $\mathsf{CC/QQ}$ separation whose concepts are not random verifiable. 
\end{proposition}

The above proposition discusses a setting where the data comes from a quantum process for two reasons.
Firstly, we have no reason to believe that the concepts in this setting are random verifiable in general.
Secondly, there are strong arguments that the concepts in this setting are probably also not in $\mathsf{HeurP}/\mathsf{poly}$ (as the converse would imply breaks in Diffie-Helman and RSA cryptosystems{\color{blue}\footnotemark[12]}).
This analysis provides a bit of information about the possibility of learning separations without using random verifiability in a few new regimes. 
However, there is one possible regime that we do not address. 
Specifically, it may be the case that there exist functions in $\big(\mathsf{HeurP/poly}\big) \cap \mathsf{BQP}$, which are not random verifiable.
Note that being outside of $\mathsf{HeurP/poly}$ guarantees that a problem is not classically efficiently learnable, but on the other hand being inside $\mathsf{HeurP/poly}$ does \emph{not} guarantee that it is classically efficiently learnable. 
Consequently, we can provide no insights for this case.
However, we  highlight that the regime of Proposition~\ref{prop:norv} arguably covers the most interesting region of learning separations for the hardest functions in $\mathsf{BQP}$, which will coincide with those from elaborate quantum experiments.

Following the reasoning above, we can reformulate Theorem~\ref{thm:recipe_separation} to also capture the learning separation of Proposition~\ref{prop:norv}.
Specifically, we can provide a ``checklist'' of criteria (albeit less intuitive than those in Theorem~\ref{thm:recipe_separation}) that when satisfied lead to a learning separation, without the examples necessarily being efficiently generatable.
Firstly, to ensure efficient quantum learnability we require that the underlying function lies in $\mathsf{BQP}$ and that the concept class put on top of this function is quantumly efficiently learnable.
Secondly, to ensure classical intractability of learning, we require that the underlying function lies outside $\mathsf{HeurP/poly}$ (i.e., the function is so hard, that even examples will not help to efficiently evaluate it classically).
We provide this ``checklist'' in the theorem below and we provide more details in Appendix~\ref{appendix:recipe2}.
We would like to highlight that the main point of this new ``checklist'' is that as long as there exist functions that lie in $\mathsf{BQP}\backslash \big(\mathsf{HeurP/poly}\big)$, then we can have \textit{both} classical intractability of learning \textit{and} efficient quantum learnability without limiting the hypothesis class to a trivial singleton class.

\begin{restatable}[Sufficient conditions for separations based on heuristic hardness -- version 2]{theorem}{recipetwo}
\label{thm:recipe_separation2}
Consider a family of concept classes $\{\mathcal{C}_n\}_{n \in \mathbb{N}}$ and distributions $\{\mathcal{D}_n\}_{n \in \mathbb{N}}$ (i.e., a learning problem).
Suppose there exists another family of concept classes $\{\mathcal{F}_n\}_{n \in \mathbb{N}}$ and a family of invertible functions $\{g_n: \mathcal{X}_n \rightarrow \mathcal{X}_n\}_{n \in \mathbb{N}}$ such that every $c \in \mathcal{C}_n$ can be decomposed as 
\begin{align}
    c(x) = f(g_n^{-1}(x)), \quad \text{for some $f \in \mathcal{F}_n$}.
\end{align}
Then, if $\{\mathcal{F}_n\}_{n \in \mathbb{N}}$ and $\{g_n\}_{n \in \mathbb{N}}$ satisfy the criteria below, the family of concept classes $\{\mathcal{C}_n\}_{n \in \mathbb{N}}$ exhibits a $\mathsf{CC}/\mathsf{QQ}$ separation under the family of input distributions $\{\mathcal{D}^g_n\}_{n \in \mathbb{N}}$, where $\mathcal{D}_n^g$ denotes the push-forward distribution $g_n(\mathcal{D}_n)$ (i.e., the distribution induced by first sampling $x \sim \mathcal{D}_n$, and then computing $g_n(x)$).

\bigskip

\noindent \underline{\textbf{Criterion 1} (hardness of $\{g_n^{-1}\}_{n \in \mathbb{N}}$)}{\color{blue}\footnotemark[13]}: 
\begin{itemize}
    \item The family of functions $\{g_n^{-1}\}_{n \in \mathbb{N}}$ lies inside $\mathsf{BQP}$ but lies outside $\mathsf{HeurP/poly}$ (with respect to input distribution $g_n(\mathcal{D}_n)$).
\end{itemize}

\noindent \underline{\textbf{Criterion 2} (efficient learnability of $\{\mathcal{F}_n\}_{n \in \mathbb{N}}$)}: 
\begin{itemize}
    \item The problem $\big(\{\mathcal{F}_n\}_{n \in \mathbb{N}}, \{\mathcal{D}_n\}_{n \in \mathbb{N}} \big)$ is in $\mathsf{QQ}$.
        That is, $\{\mathcal{F}_n\}_{n \in \mathbb{N}}$ is quantumly efficiently learnable under $\{\mathcal{D}_n\}_{n \in \mathbb{N}}$ with a classically polynomially evaluatable hypothesis class.
\end{itemize}   

\noindent \underline{\textbf{Criterion 3} (evaluating concepts $\implies$ computing $g_n^{-1}$)}:     
\begin{itemize}
    \item For every $x \in \mathcal{X}_n$, there exists subsets $\mathcal{C}'_n(x) \subset \mathcal{C}_n$, $\mathcal{X}'_n(x) \subset \mathcal{X}_n$, and a $\mathcal{O}(\mathrm{poly}(n))$-time classical algorithm $\mathcal{B}$ that maps
    \begin{align}
    \label{eq:set_g}
        \big\{c'(y) \mid c' \in \mathcal{C}'_n, \text{ }y\in \mathcal{X}'_n\big\} \mapsto g^{-1}_n(x).
    \end{align}
    \begin{itemize}
        \item In other words, if we can efficiently evaluate a subset of the concepts on a subset of points, then we can use this to efficiently invert $g_n$.
    \end{itemize}
\end{itemize}

\end{restatable}

\footnotetext[13]{Note that for Theorem~\ref{thm:recipe_separation2} to go beyond Theorem~\ref{thm:recipe_separation} we require also that evaluating $g_n$ is classically intractable (since otherwise we have random verifiability of the concepts).}

Recall that the point of Theorems~\ref{thm:recipe_separation} and~\ref{thm:recipe_separation2} is to provide a checklist of relatively easy to check criteria for a provable learning separation to streamline one's approach to proving learning separations.
Moreover, it should serve as a tool to elucidate what the bottlenecks are in proving learning separations.
To illustrate its usefulness, we will discuss four examples of how to use Theorem~\ref{thm:recipe_separation} when proving separations.
Specifically, we first show that it captures the two learning separations discussed in Section~\ref{subsec:dlp} and Section~\ref{subsec:rsa}.
Afterwards, we apply our theorem to discuss $(i)$ the importance of the ``efficient data generation'' requirement, $(ii)$ the challenges of building a learning separation from a known computational separation, and $(iii)$ the bottlenecks in establishing separations when data comes from (potentially hard to classically simulate) quantum experiment.

\subsection{Applying our checklist to (potential) separations}
\label{subsec:recipe_examples}

In this section, we discuss four didactic examples of how to use Theorem~\ref{thm:recipe_separation} and~\ref{thm:recipe_separation2} to streamline one's approach to proving learning separations, or elucidate
the bottlenecks in proving this separations.
In particular, we look at learning problems studied in literature.
The examples are as follows:
\begin{itemize}[leftmargin=*]
    \item \underline{Example 1}:\ we discuss how our theorem captures the $\mathsf{CC/QQ}$ separations of the discrete logarithm concept class (Section~\ref{subsec:dlp}) and the $\mathsf{CC/QC}$ separation of cube root concept class (Section~\ref{subsec:rsa}).
    \item \underline{Example 2}:\ we study the problem of quantum circuit evaluation, and we highlight the importance of the ``efficient data generation'' requirement. In particular, we show how having access to examples can radically enhance what a classical learner is able to efficiently evaluate.
    \item \underline{Example 3}:\ we study a learning problem based on matrix inversion and we use our checklist to elucidate the bottlenecks in building a learning separation on a computational separation. 
    \item \underline{Example 4}:\ we study Hamiltonian learning and we use our theorem to study separations when the data comes from a (potentially hard to classically simulate) quantum experiment.
\end{itemize}

\subsubsection{Example 1: the discrete logarithm and cube root concept classes}
\label{subsubsec:recipe_example_dlp-3root}

First, we show how Theorem~\ref{thm:recipe_separation} can be used to prove that the discrete logarithm concept class~\cite{liu:dlp} (Section~\ref{subsec:dlp}) and the cubic root concept class~\cite{kearns:clt} (Section~\ref{subsec:rsa}) exhibit a learning separation.

\paragraph{The discrete logarithm concept class}
Following the notation in Theorem~\ref{thm:recipe_separation}, we can decompose the discrete logarithm concept class (see Definition~\ref{def:dlp}) by considering the functions 
\begin{align}
\label{eq:fg_dlp}
     f_i(x) = \begin{cases}+1, & \text{if }x \in [i, i + \frac{p-3}{2}],\\ -1, & \text{else,} \end{cases} \quad \text{ and } \quad  g_n(x) = a^x \mod p.
\end{align}

Specifically, the concept $c_i$ defined in Eq.~\eqref{eq:c_dlp} can be written as $c_i(x) = f_i(g_n^{-1}(x))$, providing the decomposition required by Theorem~\ref{thm:recipe_separation}.
Next, we will use Theorem~\ref{thm:recipe_separation} to show that the discrete logarithm concept class exhibits a $\mathsf{CC/QQ}$ learning separation.
In order to do so, we first have to establish that the above functions satisfy the requirements in Theorem~\ref{thm:recipe_separation}.

First, we discuss the efficient data generation requirement.
It is clear that both $f_i$ and $g_n$ can be computed efficiently on a classical computer.
This automatically leads to efficient data generation, since one can uniformly random pick $x \in \{0,1\}^n$ and compute the tuple $\big(g_n(x), f_i(x)\big)$, which is a valid example for the discrete logarithm concept class.

Next, we discuss the quantum advantage requirement. 
Firstly, one can efficiently compute $g_n^{-1}$ using Shor's algorithm for the discrete logarithm~\cite{shor:factoring}.
Secondly, in~\cite{blum:hardness} it is shown that computing the most-significant bit of the discrete logarithm on a $\frac{1}{2} + \frac{1}{\mathrm{poly}(n)}$ fraction of inputs is as hard as computing the most-significant bit of the discrete logarithm on all inputs.
Moreover, it is a widely-believed conjecture that computing the most-significant bit of the discrete logarithm is classically intractable.
Putting this together, we find that under this widely-believed conjecture there does not exist a classical algorithm  with runtime $\mathcal{O}(\mathrm{poly}(n, 1/\epsilon))$ that satisfies Equation~\eqref{eq:qalg_recipe}.
Finally, one can efficiently reconstruct $x$ from $\{f_j(x)\}_{j \in J}$ where the subset $J \subset [p]$ is chosen according to a procedure similar to binary-search.
More precisely, one finds an interval containing~$x$, and afterwards one repeatedly splits this interval in half (by looking at its overlap with another interval) while keeping track of which half contains~$x$.

Having shown that the functions $\{f_i\}_{i \in [p]}$ and $g_n$ satisfy the requirements listed in Theorem~\ref{thm:recipe_separation}, we can evoke Theorem~\ref{thm:recipe_separation} and conclude that the concept class $\{\mathcal{C}_n^{\mathrm{DLP}}\}_{n \in \mathbb{N}}$ -- whose concepts can be written as $c_i(x) = f_i(g_n^{-1}(x))$ -- exhibits a $\mathsf{CC/QQ}$ learning separation.
This shows how our theorem can be applied to streamlines one's approach to proving learning separations.

\paragraph{The cube root concept class}

Following the notation in Theorem~\ref{thm:recipe_separation}, we can decompose the cube logarithm concept class (see Definition~\ref{def:cube-root}) by considering the functions 
\begin{align}
\label{eq:fg_dlp}
     f_i(x) = \text{``the $i$-th bit of $x$''}, \quad \text{ and } \quad  g_n(x) = x^3 \mod N.
\end{align}

Specifically, the concept $c_i$ defined in Eq.~\eqref{eq:c_dlp} can be written as $c_i(x) = f_i(g_n^{-1}(x))$, providing the decomposition required by Theorem~\ref{thm:recipe_separation}.
Next, we will use Theorem~\ref{thm:recipe_separation} to show that the cube root concept class exhibits a $\mathsf{CC/QC}$ learning separation.
In order to do so, we again first have to establish that the above functions satisfy the requirements in Theorem~\ref{thm:recipe_separation}.

First, we discuss the efficient data generation requirement, which is completely analogous to the discrete logarithm concept class discussed above.
In particular, it is again clear that both $f_i$ and $g_n$ can be computed efficiently on a classical computer.
This automatically leads to efficient data generation, since one can uniformly random pick $x \in \{0,1\}^n$ and compute the tuple $\big(g_n(x), f_i(x)\big)$, which in this case is a valid example for the cube root concept class.

Next, we discuss the quantum advantage requirement. 
Firstly, following the standard attack on the RSA cryptosystem, using Shor's algorithm for integer factorization~\cite{shor:factoring} one can efficiently compute a specification $[g_n^{-1}]$ that allows for an efficient classical evaluation algorithm.
In particular, one can compute the unique $d^*$ that satisfies Eq.~\eqref{eq:d}, which is a specification of $g_n^{-1}$ that allows one to efficiently evaluate $g_n^{-1}$ using a classical evaluation algorithm.
Secondly, it is a widely-believed conjecture that computing $g_n^{-1}$ without this specification is classically intractable, even on an inverse polynomial fraction of inputs~\cite{kearns:clt}.
In conclusion, we again find that under this widely-believed conjecture there does not exist a classical algorithm  with runtime $\mathcal{O}(\mathrm{poly}(n, 1/\epsilon))$ that satisfies Equation~\eqref{eq:qalg_recipe}.
Finally, it is trivial that one can efficiently reconstruct any $x \in \{0,1\}^n$ from the set $\{f_i(x)\}_{i \in[n]}$ (i.e., one learns it bit by bit).

Having shown that the functions $\{f_i\}_{i \in [n]}$ and $g_n$ satisfy the requirements listed in Theorem~\ref{thm:recipe_separation}, we can evoke Theorem~\ref{thm:recipe_separation} and conclude that the concept class $\{\mathcal{C}_n^{\mathrm{root}}\}_{n \in \mathbb{N}}$ -- whose concepts can be written as $c_i(x) = f_i(g^{-1}(x))$ -- exhibits a $\mathsf{CC/QC}$ learning separation.
This again shows how our theorem can be applied to streamlines one's approach to proving learning separations.

\subsubsection{Example 2: the power of data}
\label{subsubsec:power_data}

In this section we discuss an example that shows the importance of the ``efficient data generation'' requirement in Theorem~\ref{thm:recipe_separation}.
Specifically, we will discuss a family of functions inspired by~\cite{huang:power} that from their description alone cannot be efficiently evaluated classically, yet access to a few examples (i.e., evaluations of the function) allows a classical computer to efficiently evaluate them.

Consider a polynomial-depth parameterized quantum circuit $U(\theta, \vec{\phi})$ -- with two types of parameters $\theta \in \mathbb{R}$ parameterizing a single gate and $\vec{\phi} \in \mathbb{R}^\ell$ parameterizing multiple other gates -- that is universal in the sense that for every polynomial-depth circuit $U$ there exists parameters $\vec{\phi}^* \in \mathbb{R}^\ell$ such that 
\[
U(0, \vec{\phi}^*)\ket{0^n} = U\ket{0^n}.
\]
Moreover, assume the gates in $U$ are of the form $\mathrm{exp}\left(-\frac{i\theta}{2}A \right)$, with $A^2 = I$ (e.g., $Z$- or $X$-rotations).
By measuring the output of the circuit we define a family of single parameter functions given by
\[
f_{\vec{\phi}}(\theta) = \bra{0^n}U(\theta, \vec{\phi})^\dagger M U(\theta, \vec{\phi})\ket{0^n}.
\]

Following an argument similar to~\cite{huang:power}, due to the universality of the parameterized quantum circuit no efficient randomized classical algorithm can compute the function $f_{\vec{\phi}}$ up to constant error in time $\mathcal{O}\left(\mathrm{poly}(n)\right)$, unless $\mathsf{BPP} = \mathsf{BQP}$.
Intuitively, one might thus think that the concept class $\{f_{\vec{\phi}}\}$ exhibits a separation between classical and quantum learners.
However, it turns out that the examples given to a classical learner radically enhance what it can efficiently evaluate.
In particular, when given a few of evaluations of $f_{\vec{\phi}}$ (for some fixed but arbitrary $\vec{\phi} \in \mathbb{R}^\ell$), a classical learner is suddenly able to efficiently evaluate the function.
To see this, note that by~\cite{nakanishi:pqc} one can write the functions as
\[
f_{\vec{\phi}}(\theta) = \alpha\cos(\theta - \beta) + \gamma, \quad \text{for }\alpha, \beta, \gamma \in \mathbb{R}{\color{blue}\footnotemark[14]},
\]
where the coefficients $\alpha, \beta$ and $\gamma$ are all independent of $\theta$ (but they do depend on $\vec{\phi}$).
From this we can see that any three distinct examples $\big\{\big(\theta_i, f_{\vec{\phi}}(\theta_i)\big)\big\}_{i=1}^3$ uniquely determine $f_{\vec{\phi}}(\theta)$ and one can simply fit $\alpha, \beta$ and $\gamma$ to these three examples to learn how to evaluate $f_{\vec{\phi}}$ on unseen points.
This approach can be generalized to settings with more than one free parameter $\theta$, by using the fact that expectation values of parameterized quantum circuits can be written as a Fourier series~\cite{schuld:pqc}.
Specifically, when the number of frequencies appearing in the Fourier series is polynomial, then a polynomial number of examples suffices to fit the Fourier series and learn how to evaluate the expectation value of the quantum circuits for an arbitrary choice of parameters.

From the perspective of Theorem~\ref{thm:recipe_separation}, the above setting exemplifies the importance of the ``efficient data generation'' requirement, which eliminates the possibility that data can radically enhance what a classical learner can efficiently evaluate.
Specifically, the examples needed for a classical learner to efficiently learn the above functions are such that only a quantum computer can generate them efficiently. 
In other words, these functions exemplify how hard to generate data can radically enhance what a classical learner can efficiently evaluate.
On the other hand, from the perspective of Theorem~\ref{thm:recipe_separation2} the above fails to exhibit a separation since these functions are precisely examples of functions that are in $\mathsf{P}/\mathsf{poly}$ and which are not known to be random verifiable.
Finally, note that for certain circuits it is possible to have exponentially many terms in the Fourier series~\cite{schuld:pqc, caro:encoding}, in which case it is unclear how to classically efficiently learn it.

\subsubsection{Example 3: exploiting the hardness of matrix inversion}
\label{subsubsec:matrix_inversion}

In this section, we discuss an example of a learning problem which intuitively should have the necessary ingredients to exhibit a learning separation, yet applying Theorem~\ref{thm:recipe_separation} will reveal important challenges that need to be solved if one wants a provable separation (if it in fact were to hold).
Namely, we discuss the challenges of building a learning separation from a computational separation.

We consider the problem of quantum linear systems $A x = b$, where $A$ is a $2^n \times 2^n$-complex valued matrix, and $b$ and $x$ are both complex-valued vectors of matching dimensions. 
It is known that if $A$ is specified via a sparse-access oracle{\color{blue}\footnotemark[14]}, then a quantum computer can evaluate expectation values of a quantum state proportional to the solution $x$ in time $\mathcal{O}\left(\mathrm{poly}(n) \right)$~\cite{harrow:qls}{\color{blue}\footnotemark[15]}.
This task is known to be $\mathsf{BQP}$-complete, even when $b$ is set to the first canonical basis vector  (i.e., $|0\rangle$ in Dirac notation).
Intuitively, this problem could thus be used as a basis for a learning separation.
\footnotetext[14]{Given as input a column index $j$, a sparse-access oracle specifies the values and entries of the nonzero entries in the $j^{\text{th}}$ column.}
\footnotetext[15]{There is another dependency on the condition number of $A$ which may render the process inefficient.}

To translate this computational separation into a learning separation, one could consider the learning problem with concepts: 
\begin{align}
\label{eq:linear_concept}
    c_{i}(z) = \mathrm{sign}\big(\mathrm{Tr}\big[\mathcal{O}^{(i)}_A \ket{z}\bra{z}\big] - \frac{1}{2}\big), \quad \mathcal{O}^{(i)}_A = \big(A^{-1}\big)^T\ket{0}\bra{0}_i A^{-1}
\end{align}
where $z \in \mathbb{R}^{2^n}$, $\ket{z} \propto \sum_{i = 1}^{2^n}z_i\ket{i}$ (usually we denote the input to the concept as $x$ not $z$, but in the context of linear systems $x$ is a more natural notation for the solution of the linear system),  $\ket{0}\bra{0}_i$ corresponds to measuring the $i$th qubit in the computational basis, and $A \in \mathbb{C}^{2^n \times 2^n}$ is a sparse{\color{blue}\footnotemark[16]} Hermitian matrix specified either as a sum of local terms or by an (efficient implementation) of a sparse-access oracle{\color{blue}\footnotemark[14]}.
In other words, the concepts are defined by a family of matrices $A$ (specifying linear systems), and the evaluation of the $i$th concept $c_i$ can be realized by first solving $x = A^{-1}|0\rangle$, and then computing the norm-squared of the inner product between the solution $x$ and the input $z$.
This is exactly the specification of the $\mathsf{BQP}$-hard problem of quantum linear systems:
\begin{itemize}
    \item \textbf{Input:\ } (i) A sparse{\color{blue}\footnotemark[16]} matrix $A \in \mathbb{C}^{2^n \times 2^n}$ specified either as a sum of local terms or by an efficient implementation of a sparse-access oracle, and (ii) an $n$-qubit register prepared in the quantum state $\ket{b}$.
    \item \textbf{Output:\ }An estimate of $\bra{x}M\ket{x}$ up to $\frac{1}{\mathrm{poly}(n)}$ accuracy, where $M = \ket{0}\bra{0} \otimes I$ corresponds to measuring the first qubit and $\ket{x} \propto A^{-1}\ket{b}$,
\end{itemize}
and so could intuitively offer the possibility of a learning separation.
\footnotetext[16]{There are at most $\mathrm{poly}(n)$ nonzero entries in each \emph{column}.}

To see if whether it exhibits a separation we attempt to apply Theorem~\ref{thm:recipe_separation}. 
First, we decompose the concepts defined in Eq.~\eqref{eq:linear_concept} as $c_i(z) = f_i(g^{-1}(z))$, where
\begin{align}
\label{eq:fg_lin}
    f_{i}(z) = \mathrm{sign}\big( \mathrm{Tr}\big[\ket{0}\bra{0}_i\ket{z}\bra{z} \big]\big) \quad \text{ and } \quad g_A(z) = A\ket{z}.
\end{align}
Note that Theorem~\ref{thm:recipe_separation} requires any decomposition, and this one seems most natural.
Having chosen our decomposition, we discuss which criteria of Theorem~\ref{thm:recipe_separation} are satisfied.
Regarding efficient data generation, note that for general $z \in \mathbb{R}^{2^n}$ it is classically intractable to compute $f_{i}$ and $g_A(z)$, yet the Theorem requires this for efficient example generation.
One way to circumvent this is by considering sparse $z \in \mathbb{R}^{2^n}$ specified as a list $\{(i, z_i) \mid z_i \neq 0\}$, for which both $f_{i}$ and $g_{A}$ are efficiently computable on a classical computer.
That is, one can efficiently generate examples $\big(A\ket{z}, f_i(z)\big)$ for sparse $z \in \mathbb{R}^{2^n}$.
However, for the case where the data only consists of sparse $x$, it is unclear whether the quantum advantage criteria are met.
In particular, to the best of our knowledge it is unknown whether a classical algorithm can efficiently solve the matrix inversion problem on states of the form $A\ket{z}$ for sparse $z \in \mathbb{R}^{2^n}$.
Note that the proof of $\mathsf{BQP}$-hardness of matrix inversion requires one to apply the inverse to computational basis states (see \cite{harrow:qls}), which need not be of the form $A\ket{z}$ for sparse $z \in \mathbb{R}^{2^n}$.
Therefore, to see whether the concepts in Eq.~\eqref{eq:linear_concept} exhibit a learning separation following Theorem~\ref{thm:recipe_separation}, one first has to determine whether a classical algorithm can efficiently solve the strictly easier version of the linear systems problem where the input quantum states are of the form $\ket{b} = A\ket{z}$ for sparse $z\in \mathbb{R}^{2^n}$. 
It is possible that for this special case, efficient classical solutions exist, in which case this approach will not lead toward a separation. 
However, if the quantum linear systems problem is classically hard even under this assumption, then this would put us on the right track towards a provable learning separation.
Note that in the above discussion we have set everything up to enforce efficient example generation (i.e., random verifiability), so it is really Theorem~\ref{thm:recipe_separation} that we are concerned with (as opposed to Theorem~\ref{thm:recipe_separation2}).

In short, Theorem~\ref{thm:recipe_separation} tells us that to establish a learning separation it is important to carefully balance the efficient data generation and the quantum advantage criteria in Theorem~\ref{thm:recipe_separation}.
Specifically, taking a $\mathsf{BQP}$-complete problem and translating this into a learning separation is far from trivial, as examples can radically enhance what a classical computer can evaluate (see also Section~\ref{subsubsec:power_data}).
One way to counteract this is by ensuring that the data is efficiently generatable, but as discussed above one might lose the potential quantum advantage while doing so.

\subsubsection{Example 4: Hamiltonian learning}
\label{subsubsec:ham_learning}

In this section, we apply Theorem~\ref{thm:recipe_separation} to the problem of Hamiltonian learning.
Since for this problem the data comes from quantum experiments, it intuitively could be a good candidate for a learning separation.
We study this using Theorem~\ref{thm:recipe_separation} to elucidate bottlenecks in proving such a separation.

In Hamiltonian learning one is given measurement data from a quantum experiment, and the goal is to recover the Hamiltonian that best matches the data.
Throughout the literature, various different types of measurement data have been considered.
For example, it could be measurement data from ground states, other (non-zero) temperature thermal sates, or time-evolved states.
In our case, the data will be measurement data from time-evolved states.
In particular, we formulate Hamiltonian learning in terms of a concept class as follows.
First, we fix a (polynomially-sized) set of Hermitian operators $\{H_\ell\}_{\ell = 1}^L$.
Next, we consider a family of Hamiltonians $\{H_\beta\}_{\beta \in \mathbb{R}^L}$, where
\[
    H_\beta = \sum_{\ell = 1}^L \beta_\ell H_\ell.
\]
Finally, we define the concept class $\mathcal{C}^{\mathrm{HL}} = \{c_\beta\}_{\beta \in \mathbb{R}^L}$, with concepts defined as
\begin{align}
\label{eq:hl_concepts}
    c_\beta(z, t) = \mathrm{sign}\big(\mathrm{Tr}\big[U^\dagger(t)\rho_z U(t)O_z\big]\big), \quad U(t) = e^{it H_\beta}.
\end{align}
Here $z$ describes the experimental setup, specifying the starting state (that will evolve under $H_\beta$ for time $t$) and the observable measured at the end.
A natural specification of the concepts that a learner could output are the parameters $\beta$.
In the general case, only a quantum evaluation algorithm is able to evaluate the concepts if specified by the parameters $\beta$ (assuming $\mathsf{BQP} \not= \mathsf{BPP}$), and therefore only a $\mathsf{CC/QQ}$ separation is possible.
However, in certain restricted cases there are also efficient classical algorithms that can evaluate the concepts if specificied by the parameters $\beta$, in which case a $\mathsf{CC/QC}$ separation is possible.
Moreover, it has been shown that in certain cases a classical learner can efficiently learn $\mathcal{C}^{\mathrm{HL}}$~\cite{haah:ham_learning}, in which case no separation is possible at al.

One can look at this problem from both the perspectives of Theorem~\ref{thm:recipe_separation} and Theorem~\ref{thm:recipe_separation2}. 
The biggest challenge when trying to apply Theorem~\ref{thm:recipe_separation} to the problem of Hamiltonian learning is the ``efficient data generation'' requirement.
Specifically, this requirement is not satisfied in the general case (assuming $\mathsf{BQP} \not= \mathsf{BPP}$).
However, as discussed earlier, learning separations could very well be possible without the examples being efficiently generatable.
In particular, to construct a learning separation, it is sometimes explicitly assumed that the examples are hard to generate~\cite{brien:ham_learning}. 
More precisely, in~\cite{brien:ham_learning} the authors focused on determining whether learning is even possible at all for hard-to-simulate Hamiltonians due to the tendency of correlators in such Hamiltonians to vanish. 
The authors also suggested that a learning separation may naturally occur whenever the data itself is hard to generate.
While the data being hard to generate does imply that the method of fitting a Hamiltonian to the training data (i.e., evaluating a Hamiltonian, and adjusting it according to how well it matches the training data) is hard to perform classically, it does not necessarily imply a learning separation based on Theorem~\ref{thm:recipe_separation}.
Specifically, since the ``efficient data generation'' requirement is not satisfied, one has to worry that access to data can radically enhance what a classical learning can efficiently evaluate (see also Section~\ref{subsubsec:power_data}).
Thus, what remains to be shown for a learning separation (assuming that data generation is hard), is that a classical learner cannot evaluate the concepts in Eq.~\eqref{eq:hl_concepts}, even when having access to examples.
As mentioned earlier, one can also look at this problem from the perspective of Theorem~\ref{thm:recipe_separation2}.
When doing so, one way to argue that the criteria of Theorem~\ref{thm:recipe_separation2} are satisfied is to prove that the concept are $\mathsf{BQP}$-hard.
Specifically, if the concepts are indeed $\mathsf{BQP}$-hard, then they are likely outside $\mathsf{HeurP}/\mathsf{poly}$, in which case Theorem~\ref{thm:recipe_separation2} can imply a learning separation (without requiring efficient example generation).

\section{Conclusion}
\label{sec:conclusion}

Finding practically relevant examples of learning problems where quantum computers hold a provable edge over classical approaches is one of the key challenges for the field of quantum machine learning. 
Yet, although the field is very popular and prolific we have found no true examples (barring the cryptanalytically-motivated contrived examples we discussed), at least not in the cases where the data is classical.
One of the issues is that proving separations for learning is, as we argue, in some ways even harder than proving already challenging computational separations, and they are sensitive to subtle changes in definitions.
Moreover, it is rather cumbersome to ensure a provable separation, as many separate things need to be proven.
However, the fact that proofs of separations may be hard does not prohibit there being many important cases where a separation does exist, even if the proof eludes us. 
Note, all proofs we discussed will as a necessary ingredient require a proof that certain functions are classically heuristically hard. 
While we may conjecture this is the case for many functions, there are not many proofs known to date.
This however does not prohibit us to study the ``gray areas'', where we can at least have clear statements regarding the (simpler) conjectures and assumptions a learning separation may  potentially hinge on.

To shed some light on this challenge, we studied what made existing separations between classical and quantum learners possible, and what makes others hard to prove.
Specifically, in Theorem~\ref{thm:recipe_separation} and Theorem~\ref{thm:recipe_separation2} we distilled two sets of general and sufficient conditions (i.e., two ``checklists'') that allow one to identify whether a particular problem has the ingredients for a separation, or to elucidate bottlenecks in proving this separation.
The checklist in Theorem~\ref{thm:recipe_separation} is based on random-verifiable functions (i.e., when data can be efficiently generated) and captures the two known learning separations~\cite{liu:dlp, servedio:q_c_learnability}.
The checklist in Theorem~\ref{thm:recipe_separation2} is based on certain additional assumptions involving a number of computational complexity classes with the aim of being applicable to the setting where the data is generated by a quantum experiment.
Afterwards, to illustrate the usefulness of Theorem~\ref{thm:recipe_separation} and Theorem~\ref{thm:recipe_separation2}, we applied them to four examples of potential learning separations.
First, we showed that Theorem~\ref{thm:recipe_separation} indeed captures the learning separations of~\cite{liu:dlp} and~\cite{servedio:q_c_learnability}.
Next, we showed how our theorems can elucidate bottlenecks in proving other potential separations.
In particular, we showed that our theorems highlight the fact that data can radically enhance what a classical learner can evaluate, and how to circumvent this.
Moreover, we showed how our theorems elucidate the challenges when trying to establish learning separations for problems which intuitively should lend itself to a quantum advantage (e.g., when the concepts are built from a computational separation, or when the data comes from a quantum experiment).

In this note we discussed provable separations between classical and quantum learners in the probably approximately correct (PAC) learning framework, and it may be worthwhile to reflect on what these results may imply about the real-world applicability of quantum machine learning.
In general, there is a significant gap between the PAC framework and machine learning in practice.
One studies the hardness of abstract learning problems, whereas the other deals with the challenges of solving particular real-world learning tasks. 
Specifically, the idea of having a naturally defined meaning of ``scaling of the instance size'' is often not satisfied.
In practice, we often apply and benchmark learning algorithms on \textit{fixed-size} datasets, and often worry about generalization performances of those particular learners, based on the sizes of the training sets which are actually available. 
In contrast, in the PAC framework one is concerned with learnability in general, and the results of this note can only offer some intuition on what kinds of problems may be better suited for quantum machine learning.
However, for any fixed size real world problem, it simply may be the case that a particular classical (or quantum) learning algorithm outperforms others\footnote{Indeed, this theoretically must happen due to no-free-lunch Theorems~\cite{wolpert:nfl}}.
Further, quantum machine learning may offer (or fail to offer) other flavours of advantages over particular classical learners (such as faster training for the same hypothesis family, better preprocessing, and so on). 
At the same time, theoretical learning separations for say learning from quantum-generated data do not imply that in practice classical algorithms are outperformed by quantum algorithm given a particular size or scaling of datasets (recall, it is not meaningful to discuss whether a given dataset is exponentially or polynomially-sized without having set a notion of scaling of the instance size).
Moreover, even if such a scaling is defined, and even if we do talk about a problem which is technically not PAC learnable, then we can still apply the best algorithms we have to it and obtain \textit{some} level of performance, and the PAC framework will not tell us what happens in any such particular case.
That is to say, ``real-world machine learning'' is necessarily an experiment-driven field.
Quantum versus classical learning separations of the types we discuss in this paper are thus certainly not the end-all of discussions regarding quantum machine learning advantage.
This is especially true in this phase of quantum machine learning where we cannot yet run quantum machine learning experiments at sufficiently large sizes to convincingly compare it to state-of-the-art classical machine learning.
However, the presented type of theoretical discussion on separations may provide insights to guide us in these early days of quantum machine learning.

\paragraph{Acknowledgements} 
The authors thank Simon Marshal, Srinivasan Arunachalam, and Tom O'Brien for helpful discussions.
This work was supported by the Dutch Research Council (NWO/ OCW), as part of the Quantum Software Consortium programme (project number 024.003.037).

\bibliographystyle{alpha}
\bibliography{main}

\appendix

\section{Details regarding definitions}
\label{appendix:details}

\subsection{Constraining hypothesis classes to those that are polynomially evaluatable}
\label{appendix:poly_eval}

In this section, we discuss why one has to restrict the hypothesis class to be polynomially evaluatable in computational learning theory.
Specifically, we show that every concept class that is learnable in superpolynomial time is also learnable in polynomial time if we allow the learner to use hypotheses that run for superpolynomial time.
In other words, if we do not restrict the hypothesis class to be polynomially evaluatable, then the restriction that the learning algorithm has to run in polynomial time is vacuous (that is, it imposes no extra restrictions).
For more details we refer to~\cite{kearns:clt}.

Consider a concept class $\{\mathcal{C}\}_{n \in \mathbb{N}}$ that is learnable by a superpolynomial time learning algorithm $\mathcal{A}$ using some hypothesis class $\{\mathcal{H}\}_{n \in \mathbb{N}}$.
To show that this concept class is also learnable in polynomial time, consider the hypothesis class $\{\mathcal{H}'\}_{n \in \mathbb{N}}$ whose hypotheses are enumerated by all possible polynomially-sized sets of training examples.
Each hypothesis in $\{\mathcal{H}'\}_{n \in \mathbb{N}}$ simply runs the learning algorithm $\mathcal{A}$ on its corresponding set of examples, and it evaluates the hypothesis from $\{\mathcal{H}\}_{n \in \mathbb{N}}$ that the learning algorithm outputs based on this set of examples.
Finally, consider the polynomial-time learning algorithm $\mathcal{A}'$ that queries the example oracle a polynomial number of times and outputs the specification of the hypothesis in $\{\mathcal{H}'\}_{n \in \mathbb{N}}$ that corresponds to the obtained set of examples.
By construction, this polynomial-time learning algorithm $\mathcal{A}'$ now learns $\{\mathcal{C}\}_{n \in \mathbb{N}}$.

\subsection{Proof of Lemma~\ref{lemma:cq=qq}}

\CQequalsQQ*

\begin{proof}

Since any efficient classical algorithm can be simulated using an efficient classical algorithm it is obvious that $\mathsf{CQ} \subseteq \mathsf{QQ}$.
For the other inclusion, let $L = (\{\mathcal{C}_n\}_{n \in \mathbb{N}}, \{\mathcal{D}_n\}_{n \in \mathbb{N}}) \in \mathsf{QQ}$.
That is, the concept class $\{\mathcal{C}_n\}_{n \in \mathbb{N}}$ is learnable under distribution $\{\mathcal{D}_n\}_{n \in \mathbb{N}}$ by an efficient quantum learning algorithm $\mathcal{A}^q$ using a quantum polynomially evaluatable hypothesis class $\{\mathcal{H}_n\}_{n \in \mathbb{N}}$. 
To show that $L \in \mathsf{CQ}$, consider the quantum polynomially evaluatable hypothesis class $\{\mathcal{H}'_n\}_{n \in \mathbb{N}}$ whose hypotheses are enumerated by all possible polynomially-sized sets of training examples.
Each hypothesis in $\{\mathcal{H}'_n\}_{n \in \mathbb{N}}$ runs the quantum learning algorithm $\mathcal{A}^q$ on its corresponding set of examples, and evaluates the hypothesis that the quantum learning algorithm outputs based on the set of examples.
Finally, consider the classical polynomial-time learning algorithm $\mathcal{A}'$ that queries the example oracle a polynomial number of times and outputs the specification of the hypothesis in $\{\mathcal{H}'_n\}_{n \in \mathbb{N}}$ that corresponds to the obtained set of examples.
By construction, this classical polynomial time algorithm $\mathcal{A}'$ can learn the concept class $\{\mathcal{C}_n\}_{n \in \mathbb{N}}$ under distribution $\{\mathcal{D}_n\}_{n \in \mathbb{N}}$ using the quantum polynomially evaluatable hypothesis class $\{\mathcal{H}'_n\}_{n \in \mathbb{N}}$. 
This shows that $L \in \mathsf{CQ}$.

\end{proof}

\section{Proof of Theorem~\ref{thm:recipe_separation}}
\label{appendix:proof}

\recipe*

\begin{proof}

First, we will show that $\{\mathcal{C}\}_{n \in \mathbb{N}}$ is not efficiently classically learnable under the assumptions listed in the theorem.
To do so, we suppose that a classical learner can efficiently learn $\{\mathcal{C}\}_{n \in \mathbb{N}}$, and then show that this implies that there exists a classical algorithm that can efficiently invert $g_n$ on a $1 - 1/\mathrm{poly}(n)$ fraction of inputs (which violates one of the assumptions listed in the theorem).

Suppose a classical learner $\mathcal{A}_{\mathrm{learn}}$ can efficiently learn $\{\mathcal{C}_n\}_{n \in \mathbb{N}}$ under distribution $\{\mathcal{D}_n^g\}_{n \in \mathbb{N}}$. 
Since we can efficiently implement the example oracle $EX(c', \mathcal{D}_n^g)$ for every $c' \in \mathcal{C}'_n$\footnote{here $\mathcal{C}'_n$ is the subset described in the final bullet point of ``classical intractability''}, we can use $\mathcal{A}_{\mathrm{learn}}$ to find for every $c' \in \mathcal{C}'_n$ some hypothesis $h'$ such that 
\begin{align*}
    \mathbb{P}_{x \sim \mathcal{D}_n^g} \big[h'(x) \neq c'(x) \big] \leq \epsilon'.
\end{align*}
Recall that there exists some polynomial-time classical algorithm $\mathcal{B}$, that for every $x \in \mathcal{X}_n$ maps
\[
\{c'(y) \mid c'\in \mathcal{C}'_n(x),\text{ }y \in \mathcal{X}'_n(x) \} \mapsto g_n^{-1}(x).
\]
By making $\epsilon'$ small enough (yet still inverse-polynomial), we have for $x \sim \mathcal{D}_n^g$ the probability that 
\[
\{h'(y) \mid c'\in \mathcal{C}'_n(x),\text{ }y \in \mathcal{X}'_n(x) \} \neq
\{c'(y) \mid c'\in \mathcal{C}'_n(x),\text{ }y \in \mathcal{X}'_n(x) \}
\]
is at most some inverse-polynomial $\epsilon$.
Thus, if we apply $\mathcal{B}$ to $\{h'(y) \mid c'\in \mathcal{C}'_n(x),\text{ }y \in \mathcal{X}'_n(x) \}$, then we obtain a classical algorithm $\mathcal{A}$ that satisfies Eq.~\eqref{eq:qalg_recipe}.
Since we assumed no such classical algorithm exists, this leads a contradiction.
In other words, we can conclude that no classical learner is able to efficiently learn $\{\mathcal{C}_n\}_{n \in \mathbb{N}}$ under distribution $\{\mathcal{D}_n^g\}_{n \in \mathbb{N}}$.

Next, we will show that $\{\mathcal{C}_n\}_{n \in \mathbb{N}}$ is learnable in either $\mathsf{QQ}$ or $\mathsf{QC}$ (depending on which criteria are met).
Note that for every $c(x) = f(g_n^{-1}(x)) \in \mathcal{C}_n$, if we are given an example oracle $EX(c, \mathcal{D}^g_n)$, then using our quantum algorithm for $g^{-1}_n$ (or the efficient classical evaluation algorithm) we can turn this into an efficient example oracle $EX(f, \mathcal{D}_n)$ by applying $g_n^{-1}$ to the first coordinate of every example (i.e., we map $(x, c(x)) \mapsto (g_n^{-1}(x), c(x))$).
Next, we use this oracle together with the $\mathsf{QQ}$ or $\mathsf{QC}$ learner (depending on which criteria are met) to efficiently learn which $f$ is generating the $EX(f, \mathcal{D}_n)$ examples.
Finally, note that if we compose the hypothesis that we obtained from learning $f$ with the efficient quantum algorithm for $g^{-1}_n$ (or the efficient classical evaluation algorithm), then we obtain an hypothesis that is close enough to the hypothesis $c$ generating the $EX(c, \mathcal{D}^g_n)$ examples.
In other words, we obtain either a $\mathsf{QQ}$ or $\mathsf{QC}$ learner for $\{\mathcal{C}_n\}_{n \in \mathbb{N}}$, depending on whether we concatenate the hypothesis for $f$ with the efficient quantum algorithm for $g_n^{-1}$ or with the efficient classical evaluation algorithm (which takes as input $[g_n^{-1}]$), respectively.

\end{proof}

\section{Proof of Theorem~\ref{thm:recipe_separation2}}
\label{appendix:recipe2}

\recipetwo*

\begin{proof}[Proof-sketch]

First, we show that no classical learner can efficiently learn $\{\mathcal{C}_n\}_{n \in \mathbb{N}}$.
If a classical learner can efficiently learn $\mathcal{C}_n$, then there exist sets of examples that when given to the efficient classical learner will result in a hypothesis that is close to the concept generating the examples.
In particular, these sets of examples can be turned into as an advice string that when given to the classical learning algorithm allows it to compute the set on the right-hand side of Eq.~\eqref{eq:set_g}~\cite{huang:power}.
In other words, from Criterion 3 it now follows that the classical learning algorithm, together with the advice string consisting of the sets of examples, is able to efficiently evaluate $g_n^{-1}$.
This clearly contradicts Criterion 1, which establishes that no classical learner can efficiently learn $\{\mathcal{C}_n\}_{n \in \mathbb{N}}$.

Next, we show that there exists an efficient quantum learner that can learn $\{\mathcal{C}_n\}_{n \in \mathbb{N}}$
To see this, note that since $\{g_n^{-1}\}_{n\in \mathbb{N}}$ is in $\mathsf{BQP}$ we can efficiently turn any example $(x, c(x))$ into an example $(y, f(y))$ by computing $y = g_n^{-1}(x)$.
Afterwards, using the examples $(y, f(y))$ we apply Criterion 2 to efficiently learn $\{\mathcal{F}_n\}_{n \in \mathbb{N}}$.
Finally, we concatenate the hypothesis output by the learning algorithm for $\{\mathcal{F}_n\}_{n \in \mathbb{N}}$ with the efficient quantum algorithm for computing $\{g_n^{-1}\}_{n\in \mathbb{N}}$, which turns it into a valid hypothesis for $\{\mathcal{C}_n\}_{n \in \mathbb{N}}$ (showing that it is indeed efficiently quantum learnable).
\end{proof}
\end{document}